\documentclass[journal,draftcls,onecolumn]{IEEEtran}
\usepackage{graphicx}
\usepackage{tabularx}
\usepackage{algorithm}
\usepackage{algpseudocode}
\usepackage{booktabs}
\usepackage{makecell}
\usepackage{array}
\usepackage[caption=false,font=normalsize,labelfont=sf,textfont=sf]{subfig}
\usepackage{textcomp}
\usepackage{stfloats}
\usepackage{multirow}
\usepackage{url}
\usepackage{verbatim}
\usepackage{cite}
\usepackage{amsthm}
\usepackage{hyperref}
\usepackage{orcidlink}

\usepackage{listings}
\lstdefinestyle{ieeepy}{
  language=Python,
  basicstyle=\ttfamily\footnotesize,
  keywordstyle=\bfseries,
  breaklines=true,
  columns=fullflexible,
  keepspaces=true,
  showstringspaces=false,
  frame=none,
  numbers=none
}

\newtheorem{theorem}{Theorem}
\newtheorem{lemma}{Lemma}
\newtheorem{proposition}{Proposition}
\newtheorem{corollary}{Corollary}
\theoremstyle{definition}
\newtheorem{definition}{Definition}
\newtheorem{example}{Example}
\newtheorem{fact}{Fact}
\newtheorem{remark}{Remark}

\usepackage[
n,
operators,
advantage,
sets,
adversary,
landau,
probability,
notions,	
logic,
ff,
mm,
primitives,
events,
complexity,
asymptotics,
keys]{cryptocode}

\begin{document}
\bstctlcite{IEEEtranBSTControl}

\title{Analyzing Linear Layers in Related-Differential Cryptanalysis}

%\author{
%\IEEEauthorblockN{Yogesh Kumar,\IEEEmembership{} Akshay Ankush Yadav,~\IEEEmembership{} and Susanta Samanta,~\IEEEmembership{}}
%
%
%\thanks{This paper was produced by the IEEE Publication Technology Group. They are in Piscataway, NJ.}
%}

\author{
Yogesh Kumar~\orcidlink{0000-0001-8495-102X},
Akshay Ankush Yadav~\orcidlink{0000-0001-9591-0215},
and Susanta Samanta~\orcidlink{0000-0003-4643-5117}
%
% \thanks{Corresponding author: Susanta Samanta
% (e-mail: ssamanta@uwaterloo.ca).}
\thanks{
Yogesh Kumar and Akshay Ankush Yadav are with the Scientific Analysis Group, DRDO, Metcalfe House Complex, Delhi-110054, India
(e-mail: yogeshkumar.sag@gov.in,  akshay.ankush@gov.in).}
\thanks{
Susanta Samanta is with the Department of Electrical and Computer Engineering, University of Waterloo, Waterloo, ON N2L 3G1, Canada
(e-mail: ssamanta@uwaterloo.ca).}
}

\maketitle
\thispagestyle{plain}
\pagestyle{plain}

\begin{abstract}
In AES-like ciphers, diffusion layers are commonly instantiated using MDS matrices, since their optimal branch number yields strong diffusion guarantees and underpins classical resistance arguments against differential and linear cryptanalysis. However, Daemen and Rijmen (2009) showed that linear layers may still exhibit related-differential structure beyond what the MDS criterion captures, and Bardeh and Rijmen (2022) demonstrated that this phenomenon can be exploited in attacks on reduced-round AES. In this work, we systematically investigate the conditions under which linear layers avoid or exhibit these differentials, identifying matrix classes for which such structure is unavoidable. We first prove that every nonsingular non-MDS matrix admits a nontrivial pair of related differentials, showing that the MDS property is necessary for avoiding them. We then establish that every odd order symmetric MDS matrix admits related differentials, which rules out broad families of Cauchy-based constructions. We also substantially strengthen the circulant case by proving that related differentials are unavoidable for every circulant MDS matrix of order {$n$} with {$n \not\equiv \pm 2 \pmod{12}$}. Finally, we revisit the characterization of {$3 \times 3$} MDS matrices over {$\mathbb{F}_{2^m}$} for the absence of related differentials, and derive an explicit necessary and sufficient criterion in terms of {$15$} polynomial constraints.
\end{abstract}


\section{Introduction}
\IEEEPARstart{S}{PN}-based block ciphers represent the dominant paradigm in modern symmetric cryptography,
exemplified by the Advanced Encryption Standard (AES)~\cite{AES}. At the core of SPN design lies the wide-trail strategy~\cite{JDA_Thesis_1995}, which combines nonlinear substitution layers with a linear diffusion layer to achieve strong resistance against differential and linear cryptanalysis. In many designs, this diffusion layer is realized using an MDS matrix, since MDS matrices provide the maximum possible branch number and therefore maximize the spread of nonzero differences and linear masks across rounds. This property enables designers to derive provable lower bounds on the number of active S-boxes in differential and linear trails, forming the basis of security arguments for ciphers such as AES~\cite{AES}, SHARK~\cite{SHARK}, SQUARE~\cite{SQUARE}, and many related SPN constructions.

Motivated by the importance of MDS matrices in diffusion layer design, their construction has been studied extensively in the literature. Existing approaches can be classified broadly into two categories: non-recursive and recursive. In non-recursive constructions, the resulting matrices are directly MDS. In recursive constructions, one typically begins with a sparse matrix $A$ of order $n$, with entries chosen so that $A^n$ is MDS. Both categories can be further divided into algebraic, search-based, and hybrid methods. In algebraic constructions, the goal is to satisfy suitable algebraic conditions that guarantee the MDS property, without directly verifying the MDS property in each instance. Notable examples include Cauchy and Vandermonde based constructions, together with their generalizations~\cite{LACAN2003,kc2,kcz,Gupta2023direct,sdm} in the nonrecursive setting, and companion matrix based constructions~\cite{Augot2014ShortenedBCH,GuptaPV17_1,GuptaPV17_2} in the recursive setting. Search-based methods obtain MDS matrices by exploring constrained design spaces, often including exhaustive search, and circulant, Hadamard and related structures~\cite{Kesarwani_FSE2019,Sim2015LightweightInvolution,Liu2016GeneralizedCirculant,Li2016CirculantInvolutory} have been extensively studied in this direction. Hybrid methods~\cite{gmt,psa,Kumar_MDS2024,Kumar_4MDS} combine algebraic insight with search, typically by finding a representative MDS matrix and then using it to generate broader classes of MDS matrices.

While the above approaches primarily focus on achieving the MDS property and reducing implementation cost, it has also been observed that additional structural properties of the diffusion layer can influence cryptographic security. In particular, Daemen and Rijmen~\cite{daemen2009new} investigated linear transformations used in AES-like block ciphers in order to understand which properties of the linear layer influence the behavior of differentials and characteristics over super S-boxes. In particular, they introduced the notion of related differentials, a property of linear transformations that affects the distribution of fixed-key differential probabilities rather than the expected differential probability itself. They also showed that the AES MixColumns transformation admits related differentials. Later, Bardeh and Rijmen~\cite{GhaedBardeh_Rijmen_2022} leveraged the related differentials of the AES MixColumns transformation and combined them with the zero-difference property introduced in~\cite{RBH17}. This led to a new 7-round related-differential characteristic and, in particular, to a key-recovery attack on 7-round AES-128.

These advancements reveal that MDS matrices with the same diffusion properties need not have the same cryptographic behavior with respect to related differentials. In particular, Daemen and Rijmen~\cite{daemen2009new} showed that while $4 \times 4$ circulant MDS matrices admit related differentials, there also exist $4 \times 4$ MDS matrices without this property. More recently, Jha et al.~\cite{jha2025construction} studied this question for $4 \times 4$ Hadamard matrices and obtained a systematic construction of $4\times 4$ Hadamard MDS matrices resistant to related-differential cryptanalysis. The problem was subsequently investigated by Otal et al.~\cite{otal2025cryptographic} for $2 \times 2$ and $3 \times 3$ MDS matrices. They established that no $2 \times 2$ MDS matrix accommodates related differentials, affirmed that every $3 \times 3$ circulant MDS matrix does admit them, and proved a similar result for $3 \times 3$ involutory MDS matrices over finite fields of even characteristic. They also proposed a characterization aimed at identifying $3 \times 3$ MDS matrices over finite fields of even size that avoid related differentials.

Motivated by these results, in this paper we investigate the existence of related differentials for linear layers over finite fields. Our aim is to determine how the dimension and algebraic structure of the underlying linear layer influence the presence or absence of related differentials, and to obtain general insight into which classes of linear layers resist related-differential cryptanalysis. More specifically, our key contributions are summarized as follows:
\begin{itemize}
	\setlength{\itemsep}{1em}
	\item[$\bullet$] \textbf{Necessity of MDS property:} Although MDS matrices form a natural starting point because of their optimal diffusion and central role in block-cipher design, many lightweight block ciphers employ Near-MDS matrices to obtain a better trade-off between security and efficiency. Examples include PRIDE~\cite{PRIDE}, Midori~\cite{MIDORI}, MANTIS~\cite{SKINNY}, FIDES~\cite{Fides}, and PRINCE~\cite{PRINCE}. One well known Near-MDS example is the $4 \times 4$ circulant matrix with first row $(0,1,1,1)$, which admits related differentials~\footnote{The related differentials it has are $([1, 0, 0, 0],~[0, 1, 1, 1])$ and $([0, 0, 1, 0],~[1, 1, 0, 1])$.}. This naturally leads to the question of whether other Near-MDS matrices, or more generally matrices with smaller branch number used in the literature, also exhibit related-differential behavior. We answer this question affirmatively by proving that the MDS property is necessary for avoiding related differentials: every nonsingular non-MDS matrix admits related differentials. Therefore, the search for linear layers resistant to related-differential cryptanalysis must be confined to the MDS setting.
	
	\item[$\bullet$] \textbf{Vulnerability of symmetric matrices:} Since many cryptographically relevant matrix classes, including Hadamard-type constructions, are symmetric, it is natural to ask whether symmetry itself influences the existence of related differentials. Previous work~\cite{jha2025construction} showed that $4\times 4$ Hadamard MDS matrices do not inherently admit related differentials and derived algebraic conditions under which such matrices avoid them. This leaves open the broader question of whether symmetry as a structural property can prevent related differentials in general. We answer this question negatively by proving that all odd order symmetric MDS matrices possess related differentials. Thus, despite the prominence of symmetric matrices in cryptographic design, symmetry in odd dimension constitutes an inherent vulnerability from the viewpoint of related-differential cryptanalysis.
	
	\item[$\bullet$] \textbf{Vulnerability of circulant matrices:} Circulant matrices are among the most prominent and widely studied structured matrices in the literature, particularly because of their role in the diffusion layers of AES~\cite{AES}. Since related differentials are known to exist for all $3\times 3$ circulant MDS matrices~\cite{otal2025cryptographic} and all $4\times 4$ circulant MDS matrices~\cite{daemen2009new}, it is natural to ask whether this phenomenon extends more generally to circulant matrices. We answer this question by proving that circulant MDS matrices of order $n$ admit related differentials whenever $n \not\equiv \pm 2 \pmod{12}$, thereby yielding a general characterization of related-differential vulnerability in this important matrix family.

	\item[$\bullet$] \textbf{Complete characterization for $3\times 3$ matrices:} 
	% Jha et al.~\cite{jha2025construction} obtained a complete characterization of $4\times 4$ Hadamard MDS matrices with respect to related differentials. 
	% Otal et al.~\cite{otal2025cryptographic} studied the existence of related differentials for all $3\times 3$ MDS matrices and proposed conditions toward a characterization of this phenomenon. In this paper, we show that a complete characterization requires additional conditions. 
    We revisit the characterization of $3\times 3$ MDS matrices for the existence of related differentials  and obtain a complete algebraic characterization in terms of 15 necessary and sufficient conditions, thereby enabling exhaustive verification and the construction of matrices resistant to related differentials in this dimension. Finally, based on the characterization, we compute the number of $3\times 3$ MDS matrices that have no related differentials. Table~\ref{tab:mds_enumeration} compares the total number of $3\times 3$ MDS matrices with the number of those having no related differentials.
	
\end{itemize}

Table~\ref{tab:comparison} summarizes the scope of prior structural results on related differentials and highlights the contributions of the present work.

\begin{table*}[t]
	\centering
	\caption{Comparison of structural results on the existence of related differentials over finite fields of characteristic two}
	\label{tab:comparison}
	\renewcommand{\arraystretch}{1.1}
	\begin{tabularx}{\textwidth}{@{}l l X@{}}
		\toprule
		Work
		& Matrix class (dimension)
		& Main result \\
		\midrule

		Daemen--Rijmen~\cite{daemen2009new}
		& Circulant MDS matrices ($4\times4$)
		& Every $4\times4$ circulant MDS matrix admits related differentials. \\

		\addlinespace

		Jha et al.~\cite{jha2025construction}
		& Hadamard MDS matrices ($4\times4$)
		& Algebraic conditions and a systematic construction for $4\times4$ Hadamard MDS matrices that admit no related differentials. \\

		\addlinespace

		\multirow[t]{3}{*}{Otal et al.~\cite{otal2025cryptographic}}
        & MDS matrices ($2\times2$)
        & No $2\times2$ MDS matrix admits related differentials. \\
        % \cline{2-3}
        \addlinespace
        & \makecell[tl]{Circulant and involutory\\
       MDS matrices ($3\times3$)}
        & Every $3\times3$ circulant MDS matrix admits related differentials; the same holds for every $3\times3$ involutory MDS matrix. \\
        % \cline{2-3}
        \addlinespace
        & MDS matrices ($3\times3$)
        & Proposed conditions for identifying matrices that avoid related differentials \\

		\midrule

		\multirow[t]{4}{*}{This work}
		& Nonsingular non-MDS matrices ($n\times n$, $n\ge 2$)
		& Every nonsingular non-MDS matrix admits related differentials; hence, the MDS property is necessary for avoiding them (Theorem~\ref{thm:non-mds}). \\

		\addlinespace

		& Symmetric MDS matrices ($n\times n$, odd $n\ge 3$)
		& Every symmetric MDS matrix of this order admits related differentials (Theorem~\ref{thm:symmetric}). \\

		\addlinespace

		& Circulant MDS matrices ($n\times n$, $n\geq 3$ and $n\not\equiv\pm2\pmod{12}$)
		& Every circulant MDS matrix of such order admits related differentials (Theorem~\ref{thm:circulant}). \\

		\addlinespace

		& MDS matrices ($3\times3$)
		& A complete characterization through 15 necessary and sufficient polynomial conditions for avoiding related differentials (Theorem~\ref{thm3*3}). \\

		\bottomrule
	\end{tabularx}
\end{table*}

\vspace{1em}

\noindent \textbf{Cryptographic motivation. } The results above have direct implications for designing diffusion layers in AES-like ciphers. Bardeh and Rijmen~\cite{GhaedBardeh_Rijmen_2022} exploited related differentials of the AES MixColumns matrix, together with the zero-difference property of~\cite{RBH17}, to derive a seven-round related-differential characteristic and a key-recovery attack on 7-round AES-128. More recently, Yan et al.~\cite{Yan_RD_2026} showed that the related-differential structure of the AES linear layer can also be exploited through exchanged-diagonal related differentials. By combining a one-round exchanged-diagonal related differential with a four-round related-differential structure, they obtained a non-random property of 7-round AES. These attacks demonstrate that related differentials of the diffusion layer can have concrete cryptanalytic consequences. This raises a design question. \emph{Which diffusion matrices admit no nontrivial related differentials, and can such matrices be characterized algebraically?}

Our results show that the existence of related differentials is often forced by algebraic properties of the linear layer, rather than being a design artifact of the AES MDS matrix. In particular, every nonsingular non-MDS matrix, every symmetric MDS matrix of odd order, and circulant MDS matrices in most practically relevant dimensions admit related differentials. Thus, a designer seeking resistance to this type of analysis should restrict attention to MDS matrices and should avoid symmetric and circulant constructions whenever the corresponding impossibility results apply. On the constructive side, our characterization of $3\times 3$ MDS matrices identifies precisely which matrices admit no related differentials. While we do not introduce a new attack, our results provide algebraic criteria for identifying and excluding related-differential structure at the design stage.

\vspace{1em}
\noindent The remainder of this paper is organized as follows. Section~\ref{Sec:Definition} presents preliminaries on MDS matrices and related differentials. Section~\ref{Sec:RDvsMDS} develops the theory of related differentials for non-MDS and symmetric matrices. Section~\ref{Sec:RDvsCirculant} examines related differentials in circulant MDS matrices. Section~\ref{Sec:RD_Charac} gives a complete characterization for $3\times 3$ matrices. Finally, Section~\ref{Sec:Conclusion} concludes the paper.

\section{Mathematical Preliminaries}~\label{Sec:Definition}
In this section, we collect the notation and basic definitions used throughout the paper. We work over finite fields of characteristic $2$ and recall the notions of branch number, MDS matrices, and related differentials that will be used in the subsequent sections.

Throughout this paper, $\mathbb{F}_{2^m}$ denotes the finite field with $2^m$ elements. We denote the multiplicative group of the finite field $\mathbb{F}_{2^m}$ by $\mathbb{F}_{2^m}^*$, and the set of vectors of length $n$ with entries from the finite field $\mathbb{F}_{2^m}$ is denoted by $\mathbb{F}_{2^m}^n$. The Hamming weight $\mathrm{wt}(x)$ of a vector $x\in \mathbb{F}_{2^m}^n$ is the number of nonzero components: $\mathrm{wt}(x)=|\set{i:x_i\neq 0}|$.

The security of SPN block ciphers against differential and linear cryptanalysis depends critically on the branch number~\cite{JDA_Thesis_1995} of the linear diffusion layer. For a linear transformation $M:\mathbb{F}_{2^m}^n \rightarrow \mathbb{F}_{2^m}^n$, the differential branch number is defined as
\[
B_d(M)=\min_{x\neq 0}\bigl(\mathrm{wt}(x)+\mathrm{wt}(Mx)\bigr),
\]
and the linear branch number is defined as
\[
B_l(M)=\min_{x\neq 0}\bigl(\mathrm{wt}(x)+\mathrm{wt}(M^Tx)\bigr).
\]
A linear transformation $M$ is called \emph{Maximum Distance Separable (MDS)} if its branch number attains the theoretical maximum, namely 
\[B_d(M)=B_l(M)=n+1. \]

MDS matrix finds its practical applications as a diffusion layer in cryptographic primitives. The concept of the MDS matrix comes from coding theory, specifically from the realm of maximum distance separable (MDS) codes. An $[n, k, d]$ code is MDS if it meets the singleton bound $d = n-k + 1$.

\begin{theorem}~\cite[page 321]{FJ77} 
	An $[n, k, d]$ code $C$ with a generator matrix $G = [ I ~|~ M ]$, where $M$ is a $k \times ( n - k )$ matrix, is MDS if and only if every square submatrix (formed from any $i$ rows and any $i$ columns, for any $i = 1, 2,\ldots, \min \{k, n - k \}$) of $M$ is nonsingular.
\end{theorem}

The following fact is another way to characterize an MDS matrix.

\begin{fact}
	A square matrix $M$ is an MDS matrix if and only if every square submatrix of $M$ is nonsingular. 
\end{fact}

\noindent We now recall the notions of related differences and related differentials introduced in~\cite{daemen2009new}.

Let $M: \mathbb{F}_{2^m}^{n}\rightarrow \mathbb{F}_{2^m}^{n}$ be a linear map. For an input difference $u\in \mathbb{F}_{2^m}^n$, the corresponding output difference is $v=Mu$. The pair $(u,v)$ is called a \emph{differential} of the linear map $M$.

% Given an input difference vector $u$ and a MixColumns matrix $M$, let
% \[
% v = Mu
% \]
% denote the corresponding output difference. The input/output difference pair $(u, v)$ is called a differential.  

We restate the definitions of related differences and related differentials from~\cite{daemen2009new}.

\begin{definition} \label{related_diff}
	Two vectors 
	\[
	u= [u_1 , u_2 , \dots , u_n]^T
	\quad \text{and} \quad 
	v= [v_1 , v_2 , \dots , v_n]^T
	\]
	are called \emph{related differences} if and only if
	\[
	u_i \cdot v_i \cdot (u_i \oplus v_i) = 0,
	\quad ~\forall ~i \in \{1,2,\dots,n\}.
	\]
\end{definition}

Equivalently, for each coordinate $i$, at least one of the three values $u_i$, $v_i$, and $u_i \oplus v_i$ is zero.

\begin{definition}
	Two differentials $(u, Mu)$ and $(v, Mv)$ for a linear map are called \emph{related differentials} if and only if $(u, v)$ and $(Mu, Mv)$ are related differences.
\end{definition}

\noindent \textbf{Observation: }
Let $u,v\in\mathbb{F}_{2^m}^{\,n}$ be two vectors such that $u_i\neq 0$ and $v_i\neq 0$ for every $i\in\{1,\ldots,n\}$. Equivalently, $\mathrm{wt}(u)=\mathrm{wt}(v)=n$. Suppose that $u$ and $v$ are related differences. By Definition~\ref{related_diff}, for every coordinate $i$, at least one of $u_i$, $v_i$, and $u_i\oplus v_i$ must be zero. Since $u_i\neq 0$ and $v_i\neq 0$, it follows that $u_i\oplus v_i=0$. Hence, $u_i=v_i$ for every $i$, and therefore $u=v$. Thus, two distinct vectors whose coordinates are all nonzero cannot form a pair of related differences. This observation is useful for determining the admissible Hamming weight patterns of the input and output differences considered in the proof of Theorem~\ref{thm3*3}.

\begin{example}
	Consider the matrix
	\[M= \begin{pmatrix}
		0 & 1 & 1 & 1\\
		1 & 0 & 1 & 1\\
		1 & 1 & 0 & 1\\
		1 & 1 & 1 & 0
	\end{pmatrix} \]
	over the finite field $\mathbb{F}_{2^m}$. Let
	\[
	u=[1,0,0,0]^T, \qquad v=[0,0,1,0]^T.
	\]
	Then
	\[
	Mu=[0,1,1,1]^T, \qquad Mv=[1,1,0,1]^T.
	\]
	Hence, the differentials $(u,Mu)$ and $(v,Mv)$ are related differentials for the linear map defined by $M$.
\end{example}

\begin{remark}\label{remark-trivial}
We use the following terminology to distinguish between trivial and nontrivial related differences and related differentials.
\begin{enumerate}
    \item Two related differences $u$ and $v$ are called \emph{trivial} if at least one of $u$, $v$, or $u \oplus v$ is the zero vector. Equivalently, they are trivial if $u=0$, $v=0$, or $u=v$. Otherwise, $(u,v)$ is called \emph{nontrivial}.

    \item Let $M$ be an $n \times n$ matrix over $\mathbb{F}_{2^m}$ and $(u,Mu)$ and $(v,Mv)$ be two related differentials of $M$. We call them \emph{trivial related differentials} if either the input related-difference pair $(u,v)$ or the output related-difference pair $(Mu,Mv)$ is trivial. Equivalently, at least one of
    \[
    \begin{aligned}
     & u=0,\quad v=0,\quad u=v,\\
     & Mu=0,\quad Mv=0,\quad \text{or}\quad Mu=Mv
    \end{aligned}
    \]
    % \[
    % u=0,\quad v=0,\quad u=v,\quad
    % Mu=0,\quad Mv=0,\quad \text{or}\quad Mu=Mv
    % \]
    holds. Otherwise, they are called \emph{nontrivial related differentials}.
    \item Let $M$ be a nonsingular $n\times n$ matrix over $\mathbb{F}_{2^m}$, and $(u,Mu)$ and $(v,Mv)$ be two related differentials of $M$ with $u\ne0$, $v\ne0$, and $u\ne v$. Then $Mu\ne0$, $Mv\ne0$, and $Mu\ne Mv$, so $(u,Mu)$ and $(v,Mv)$ form a nontrivial related differential pair.
\end{enumerate}
\end{remark}

\begin{remark}
    Throughout the remainder of this paper, we consider only nontrivial related differential pairs. Accordingly, we exclude cases in which any of $u$, $v$, or $t=u\oplus v$ is the zero vector, since these correspond to trivial related differences and yield only trivial related differential pairs.
\end{remark}

\begin{remark}
It is worth noting that Otal et al.~\cite{otal2025cryptographic} proved that no $2\times 2$ MDS matrix admits a nontrivial related differential pair. Accordingly, our investigation of related differentials in MDS matrices focuses on orders $n\ge3$. Theorem~\ref{thm:non-mds} nevertheless also applies to nonsingular non-MDS matrices of order two.
\end{remark}

\noindent In the next section, we present our main structural results on the classes of matrices that admit related differentials.

\section{Related Differentials and MDS Matrices}\label{Sec:RDvsMDS}
This section develops two fundamental theorems: first, that every nonsingular non-MDS matrix always admits related differentials; second, that odd order symmetric MDS matrices necessarily possess related differentials. These results provide both necessary conditions for resistance and identify specific vulnerable construction classes.

We first begin by showing that any nonsingular non-MDS matrix contains related differentials. Let $M=[a_{i,j}]_{n\times n}$ be a square matrix of order $n$. Any submatrix of $M$ of order $1\leq r \leq n$ is constructed by selecting a subset of its $r$ rows and $r$ columns. It can be denoted as $M[i_1,i_2,\dots,i_r; j_1,j_2,\dots, j_r]$, where $i_t$ represents rows selected and $j_t$ represents columns selected of matrix $M$ for $1\leq t \leq r$. To further simplify, let $I=\{i_1,i_2 \dots, i_r\}$ and $J=\{j_1,j_2,\dots ,j_r\}$ be subsets of rows and columns of matrix $M$, respectively. Then $M[I;J]$ is a submatrix of $M$ of order $r$. Furthermore, the complementary submatrix of $M[I;J]$ is obtained by removing $I$ rows and $J$ columns of $M$ and it is denoted by $M[I^c;J^c]$, where $I^c= \{1,2,\dots, n\}\backslash I$ and $J^c= \{1,2,\dots, n\}\backslash J$.
The following theorem shows that a nonsingular non-MDS matrix always has a related differential pair. This means the MDS property is necessary to avoid related differentials. This also emphasizes the importance of MDS matrices in cryptography.

\begin{theorem}\label{thm:non-mds}    
    Let $M=[a_{i,j}]_{n\times n}$ be a nonsingular matrix of order $n$. If $M$ is not MDS, then $M$ has a related differential pair.
\end{theorem}
\begin{proof}
	Suppose that $M$ is not an MDS matrix. Since $M$ is nonsingular, we have $\det(M)\neq 0$. Then there exists a $r\times r$ submatrix $M[I;J]$ of $M$ such that $r<n $ whose determinant is zero, where $I=\{i_1,\dots, i_r\}$ and $J=\{j_1,\dots ,j_r\}$ are selected rows and columns of $M$. Thus, we have a nonzero vector $s$ of length $r$, say, $s=[s_1, \dots, s_r]^T$ such that 
	\[M[I;J]s= 0.\] 
	
	Also observe that, for any scalar $\alpha$, $\alpha s$ will also go to zero under the map $M[I;J]$. Now, define vector $u= [u_1, \dots, u_n]^T$ of length $n$ such that 
% \[
%	u_j=\left\{
%	\begin{aligned}
%		\alpha s_t \  \text{if} \ j=j_t,\\
%		0 \ \text{ elsewhere}
%	\end{aligned}
%	\right.
%	.\]

\[
u_j = \left\{
\begin{aligned}
  &\alpha s_t && \text{if } j = j_t,\\
  &0          && \text{elsewhere}
\end{aligned}
\right.
.\]
	
	It is observed that the vector $Mu= u'=[u'_1,u'_2,\dots, u'_n]^T$ will be given as \[u'_i= \sum_{t=1}^r\alpha a_{i,j_t}s_{t} \text{ for all } 1\leq i \leq n.\] 
	
	Thus, we have for $i \in I$, $u'_i=0$. Now, we have the following two cases:
	
	\noindent\textbf{Case 1:} $\mathbf{\det(M[I^c;J^c])=0}$.
	
	If $\det(M[I^c;J^c])=0$, then, we have a nonzero vector $w=[w_1,\dots,w_{n-r}]^T$ such that \[M[I^c;J^c]w=0.\] 
	
	Let $J^c=\{k_1,k_2, \dots, k_{n-r} \}$ and define a vector $v=[v_1, \dots, v_n]^T$ such that 
	%\[v_j=\left\{
	%\begin{aligned}
	%	w_{t} \  \text{if} \ j= k_t,\\
	%	0 \ \text{ if } \ j\neq k_t 
	%\end{aligned}
	%\right.
	%\text{ for } 1\leq t \leq n-r.\] \\

\[
v_j = \left\{
\begin{aligned}
  &w_t && \text{if } j = k_t,\\
  &0   && \text{if } j \neq k_t
\end{aligned}
\right.
\qquad \text{for } 1 \leq t \leq n-r.
\] 
	
	Thus, we will have \[Mv=v'=[v'_1,v'_2, \dots,v'_n]^T,\] 
	where $v'_i= \sum_{1\leq t \leq n-r} a_{i,k_t}w_{t}$. \\
	
	Observe that if $i \notin I$, $v'_i=0$. Thus, $(u',v')$ forms related differences. Also $(u,v)$ forms related differences as $u_j=0$ if $j \not\in J$ and $v_j=0$ if $j\in J$. Therefore, this choice of $(u, Mu)$ and $(v,Mv)$ will form a related differential pair.
	
	\noindent\textbf{Case 2:} $\mathbf{\det(M[I^c;J^c]) \neq0}$.
	
	\noindent Now, consider the case when, $det(M[I^c;J^c])\neq0$ then, we know that 
	\[M[I^c;J^c]X = 0\] will have only a trivial solution. \\
	
	Now, to construct a related differential pair, let $X=[x_1, \dots, x_n]^T$ be a general vector (variables). Define a vector $v=[v_1, \dots, v_n]^T$ such that 
%	\[v_j=\left\{
%	\begin{aligned}
%		x_j \ \text{if} \ j \in J^c, \\
%		0 \ \text{ if }\ j \in J \
%	\end{aligned}
%	\right.
%	.\] \\

\[
v_j = \left\{
\begin{aligned}
  &x_j && \text{if } j \in J^c,\\
  &0   && \text{if } j \in J
\end{aligned}
\right.
.\]
	Then, we have \[Mv=v'=[v'_1,v'_2, \dots,v'_n]^T,\] 
	where $v'_i= \sum_{1\leq j \leq n, j \in J^c} a_{i,j}x_{j}$.\\
	
	As we already know, $u$ and $v$ form related differences.
	Thus, $(u, u')$ and $(v,v')$ will form related differentials if $u'$ and $v'$ form related differences. \\
	
	It is observed that if $i \in I$, $u'_i=0$. Let us assume that $u'_i\oplus v'_i=0$ for $i \notin I$. Then 
	\[v'_i= u'_i \implies \sum_{1\leq j \leq n, j \in J^c} a_{i,j}x_{j}= \sum_{t=1}^r\alpha a_{i,j_t}s_{t}.\]
	Thus we have the following system of $(n-r)$ homogeneous equations:
	
	\[ \sum_{1\leq j \leq n, j \in J^c} a_{i,j}x_{j}+ \left(\sum_{t=1}^r a_{i,j_t}s_{t}\right)\alpha=0.\]
	
	The above system has $n-r+1$ variables, namely $x_j$ and $\alpha$. Thus, this system of equations will definitely have a nontrivial solution. Since $\det(M[I^c;J^c]) \neq 0$, $\alpha$ will not be zero. The solution will provide related differentials $(u, u')$ and $(v,v')$.
\end{proof}

\noindent Thus, we have the following corollary.

\begin{corollary}\label{Coro:nonMDShasRD}
    If a nonsingular square matrix $M$ has no related differential pair, then $M$ is an MDS matrix.
\end{corollary}

\begin{remark}
    The nonsingularity assumption in Theorem~\ref{thm:non-mds} and Corollary~\ref{Coro:nonMDShasRD} is natural in the cryptographic setting. The linear layer of an SPN must be invertible to permit decryption, so only nonsingular matrices are of practical interest as diffusion layers. Moreover, every MDS matrix is nonsingular. Hence this restriction incurs no loss for the results on MDS matrices that follow.
\end{remark}

\noindent It is well known that the MDS property is preserved under matrix inversion. Corollary~\ref{Coro:nonMDShasRD} shows that matrices admitting no related differential pair form a subclass of the MDS matrices. This naturally raises the question of whether the absence of related differentials is likewise preserved under inversion. The following proposition shows that the existence of related differentials is preserved under matrix inversion.

\begin{proposition}
    Let $M$ be an invertible square matrix of order $n$. Then $M$ has a related differential pair if and only if $M^{-1}$ has a related differential pair.
\end{proposition}

\begin{proof}
	Let us assume that $M^{-1}$ has a related differential pair, say $(u,M^{-1}u), (v, M^{-1}v)$. Since $M$ is invertible, it is a bijective map. Thus, we have vectors $s$ and $t$ such that $Ms=u$ and $Mt=v$. Therefore $(M^{-1}u,u)=(s,Ms)$ and $(M^{-1}v,v)=(t,Mt)$ form a related differential pair for matrix $M$.	In the same manner, one can show that if $M$ has a related differential pair, then $M^{-1}$  also has a related differential pair.
\end{proof}

\noindent Although the MDS property is necessary for avoiding related differentials, it is not sufficient. Specifically, we will prove in Theorem~\ref{thm:symmetric} that every symmetric MDS matrix of odd order admits a related differential pair. We first establish the following auxiliary result, which will be used in its proof.

\begin{proposition}\label{prop:odd-zero}
Let $A=[a_{i,j}]_{n\times n}$ be a symmetric matrix of odd order $n$ over $\mathbb{F}_{2^m}$. If all diagonal entries of $A$ are zero, then $\det(A)=0.$
\end{proposition}

\begin{proof}
Let $n=2k+1$. By the Leibniz formula, we have
$\det(A)=\sum_{\sigma\in S_n}
\operatorname{sgn}(\sigma)
\prod_{i=1}^n a_{i,\sigma(i)},
$
where $\operatorname{sgn}$ is the sign function of permutations in the permutation group $S_n$, which returns $+1$ and $-1$ for even and odd permutations, respectively.
For $\sigma\in S_n$, define
\[
T_\sigma =
\operatorname{sgn}(\sigma)
\prod_{i=1}^n a_{i,\sigma(i)}.
\]
Since $A$ is symmetric, we have $a_{i,j}=a_{j,i}.$ Moreover, for every $i \in \{1,2,\dots,n\}$ there exists $j\in \{1,2,\dots,n\}$ such that $\sigma(i)=j$. 
Therefore,
\begin{eqnarray*}
\displaystyle \prod_{i=1}^n a_{i,\sigma(i)}  &= \displaystyle \prod_{i=1}^n a_{\sigma(i),i}  &= \displaystyle \prod_{j=1}^n a_{j,\sigma^{-1}(j)}.\\
\end{eqnarray*}
Also,
$
\operatorname{sgn}(\sigma^{-1}) = 
\operatorname{sgn}(\sigma).
$
Consequently,
$
T_\sigma=T_{\sigma^{-1}}.
$

Now partition $S_n$ into pairs $ \{\sigma,\sigma^{-1}\}$ for $\sigma\neq\sigma^{-1}$, together with the involutions satisfying $\sigma=\sigma^{-1}.$

If $\sigma\neq\sigma^{-1}$, the two terms indexed by $\sigma$ and $\sigma^{-1}$ satisfy
\[
T_\sigma+T_{\sigma^{-1}}= 2T_\sigma = 0. 
\]
Thus, all terms corresponding to permutations with
$\sigma\neq\sigma^{-1}$ cancel in pairs.

It remains to consider the permutations satisfying
$\sigma=\sigma^{-1}$. Such a permutation is an involution and therefore is a product of disjoint transpositions and fixed points. Since $n$ is odd, every such permutation has at least one fixed point. Hence, there exists some index $i$ such that
$\sigma(i)=i$. The corresponding term $T_\sigma$ then contains the factor $a_{i,\sigma(i)}=a_{i,i}.$ Since all diagonal entries of $A$ are zero, $a_{i,i}=0$, and therefore $T_\sigma=0.$

Therefore, every term in the Leibniz expansion either cancels with an identical term corresponding to $\sigma^{-1}$ or is itself zero. Hence $\det(A)=0$. This completes the proof.
\end{proof}

\begin{theorem}\label{thm:symmetric}
	Let $M$ be a symmetric MDS matrix of order $n$ over $\mathbb{F}_{2^m}$. If $n$ is odd, then $M$ admits a related differential pair. 
\end{theorem}
\begin{proof}
	Let $M=[a_{i,j}]_{n\times n}$ be a symmetric MDS matrix. Then, we have 
	\[a_{i,j}=a_{j,i} ~ \forall ~ 1\leq i,j \leq n. \] 
	
	If $n$ is odd, then $n+1$ is an even number. Let us consider 
    \[
    \begin{aligned}
     &  u=[u_1,u_2,\dots, u_{\frac{n+1}{2}},0,0,\dots,0]^T \text{ and }  v=[0,0,\dots,0, u_{\frac{n+1}{2}},u_{\frac{n+1}{2}+1}, \dots,u_n]^T,
    \end{aligned}
    \]
    % \[u=[u_1,u_2,\cdots, u_{\frac{n+1}{2}},0,0,\cdots,0]^T \text{ and }  v=[0,0,\cdots,0, u_{\frac{n+1}{2}},u_{\frac{n+1}{2}+1}, \cdots,u_n]^T,\]
	where $u_i$'s are variables taking values in the finite field \(\mathbb{F}_{2^m}\). 
	Observe that $u$ and $v$ form a related difference pair. Now, let $Mu=w$ and $Mv=s$. Then from Remark~\ref{remark-trivial}(3), it is enough to find $u$ and $v$ such that $w$ and $s$ will be a related difference pair.
	
	Now, divide the matrix $M$ into two matrices by splitting the columns of $M$ in two parts: 
	\[J_1=\{1,2,\dots, \frac{n+1}{2}\} \text{ and } J_2=\{\frac{n+1}{2},\frac{n+1}{2}+1,\dots, n\}.\] 
	
	We have, $|J_1|=|J_2|=\frac{n+1}{2}$. Furthermore, we have 
	\[Mu=M[:;J_1][u_1, \dots, u_{\frac{n+1}{2}}]^T= [w_1,w_2, \dots, w_n]^T\] and \[Mv=M[:;J_2][u_{\frac{n+1}{2}}, \dots, u_n]^T=[s_1,s_2,\dots,s_n]^T.\] 
	
	We have total $n$ variables $u_1,\dots,u_n$. Now divide the rows into three parts:  
    \[
    \begin{aligned}
     & I_1=\{1,2,\dots, \frac{n-1}{2}\} ,~ \{\frac{n+1}{2}\}, \text{ and } I_2=\{\frac{n+1}{2}+1,\frac{n+1}{2}+2,\dots, n\}.
    \end{aligned}
    \]
    % \[I_1=\{1,2,\dots, \frac{n-1}{2}\},~ \{\frac{n+1}{2}\}, \text{ and } I_2=\{\frac{n+1}{2}+1,\frac{n+1}{2}+2,\dots, n\}  .\]
	
	Then, set $w_i=0$ for $\frac{n+1}{2} < i \leq n$. We get
	\begin{equation}\label{thm4:1}
		M[I_2;J_1][u_1,u_2,\dots,u_{\frac{n+1}{2}}]^T=0.
	\end{equation}

	It has $\frac{n-1}{2}$ equations and $\frac{n+1}{2}$ variables. Thus, it has nontrivial solutions; using the Gaussian elimination method, we can compute values of $u_i$ in terms of $u_{\frac{n+1}{2}}$ for $1 \leq i \leq \frac{n-1}{2}$.\\
	
	Similarly, set $s_i=0$ for $1 \leq i < \frac{n+1}{2}$. We get
	\begin{equation}\label{thm4:2}
		M[I_1;J_2][u_{\frac{n+1}{2}}, \dots, u_n]^T=0. 
	\end{equation}
	
	It also has $\frac{n-1}{2}$ equations and $\frac{n+1}{2}$ variables. Again, it has nontrivial solutions; we can compute it in terms of $u_{\frac{n+1}{2}}$.\\
	
	Thus, we will have 
    \[
    \begin{aligned}
      & Mu=[w_1,w_2, \ldots, w_{\frac{n+1}{2}},0,0,\ldots,0]^T \text{ and } Mv=[0,0,\ldots,0, s_{\frac{n+1}{2}}, \ldots, s_n]^T.
    \end{aligned}
    \]
	% \[Mu=[w_1,w_2, \dots, w_{\frac{n+1}{2}},0,0,\dots,0]^T \text{ and } Mv=[0,0,\dots,0, s_{\frac{n+1}{2}}, \dots, s_n]^T.\] 
	i.e., $(Mv)_i=0 ~ \forall ~ 1\leq i \leq \frac{n-1}{2} $ and $(Mu)_i=0 ~ \forall ~\frac{n+1}{2}+1 \leq i \leq n$.
	Therefore, $Mu$ and $Mv$ will form a related difference pair if $s_{\frac{n+1}{2}}=w_{\frac{n+1}{2}}$ i.e. 
	
	\begin{equation}\label{thm4:3}
		\sum_{\frac{n+1}{2}\leq j \leq n} a_{\frac{n+1}{2}, j} u_j=\sum_{1\leq j \leq \frac{n+1}{2}} a_{\frac{n+1}{2}, j} u_j.
	\end{equation}

	Thus, if a nontrivial solution $[u_1,u_2, \dots, u_n]$ to the above Equations (\ref{thm4:1}), (\ref{thm4:2}), and (\ref{thm4:3}) exists, then $Mu$ and $Mv$ will form a related difference pair.
	We have 
	\[
    \resizebox{0.6\textwidth}{!}{$\displaystyle
    \begin{pmatrix}
		0_{\frac{n-1}{2}\times\frac{n-1}{2}}& &M[I_1;J_2]\\
		&&\\
		M[\{\frac{n+1}{2}\};J_1\backslash\{\frac{n+1}{2}\}] & \ \  0 \ \ &M[\{\frac{n+1}{2}\};J_2\backslash\{\frac{n+1}{2}\}] \\   
		&&\\M[I_2;J_1]& & 0_{\frac{n-1}{2}\times\frac{n-1}{2}}\\
	\end{pmatrix} 
	\begin{pmatrix}
		u_1\\ u_2 \\ \vdots \\ u_{n-1} \\ u_n
	\end{pmatrix} = 0. 
    $}
    \]
	
	Therefore, if the determinant  of the matrix  
    \[
    \resizebox{0.6\textwidth}{!}{$\displaystyle
    M'= \begin{pmatrix}
		0_{\frac{n-1}{2}\times\frac{n-1}{2}}& &M[I_1;J_2]\\
		&&\\
		M[\{\frac{n+1}{2}\};J_1\backslash\{\frac{n+1}{2}\}] & \ \  0 \ \ &M[\{\frac{n+1}{2}\};J_2\backslash\{\frac{n+1}{2}\}] \\   
		&&\\M[I_2;J_1]& & 0_{\frac{n-1}{2}\times\frac{n-1}{2}}\\
	\end{pmatrix}
    $}
    \]
 %    \[M'= \begin{pmatrix}
	% 	0_{\frac{n-1}{2}\times\frac{n-1}{2}}& &M[I_1;J_2]\\
	% 	&&\\
	% 	M[\{\frac{n+1}{2}\};J_1\backslash\{\frac{n+1}{2}\}] & \ \  0 \ \ &M[\{\frac{n+1}{2}\};J_2\backslash\{\frac{n+1}{2}\}] \\   
	% 	&&\\M[I_2;J_1]& & 0_{\frac{n-1}{2}\times\frac{n-1}{2}}\\
	% \end{pmatrix}\]
	is zero, then we will get a related differential pair.

   Observe that $M'$ is a symmetric matrix of odd order having zeroes in the diagonal entries. Thus, using Proposition \ref{prop:odd-zero}, we have $\det(M')=0$. Therefore, every symmetric MDS matrix of odd order over $\mathbb{F}_{2^m}$ has a related differential pair of the form constructed above.
\end{proof}

\noindent Thus, although symmetric MDS matrices play an important role in cryptographic design, Theorem~\ref{thm:symmetric} shows that symmetry in MDS matrices of odd dimension is an inherent vulnerability from the viewpoint of related-differential cryptanalysis. Motivated by this observation, we recall the well known Cauchy-based construction, which provides a concrete class of MDS matrices vulnerable to related-differential cryptanalysis.

Given $\{x_0, x_1, \ldots, x_{n-1}\} \subseteq \mathbb{F}_{2^m}$ and $\{y_0, y_1, \ldots, y_{n-1}\} \subseteq \mathbb{F}_{2^m}$ such that $x_i+y_j \neq 0$ for all $0 \le i,j \le n-1$, the matrix $M=(a_{i,j})$, where $a_{i,j}=\frac{1}{x_i+y_j}$, is called a Cauchy matrix. Cauchy-based constructions of MDS matrices appear in several forms in the literature, depending on how the elements $x_i$ and $y_i$ are chosen~\cite{kcz}. A particularly important special case occurs when $y_i=l+x_i$ for some nonzero $l\in\mathbb{F}_{2^m}$. Further constructions are obtained when the $x_i$'s are taken from an additive subgroup of $\mathbb{F}_{2^m}$, and a Hadamard MDS matrix of order $n=2^k$ arises when this subgroup is the linear span of $k$ linearly independent elements. We refer to these as Cauchy-based constructions of Types 2, 3, and 4, respectively~\cite{kcz}.

For the Cauchy-based construction of Type 2, where $y_i=l+x_i$ for some nonzero $l\in\mathbb{F}_{2^m}$, we have
\[
a_{i,j}=\frac{1}{x_i+y_j}=\frac{1}{l+x_i+x_j}=\frac{1}{l+x_j+x_i}=a_{j,i}.
\]
Therefore, the resulting matrix is symmetric. Theorem~\ref{thm:symmetric} therefore yields the following corollary.

\begin{corollary}\label{Coro:Type2Cauchy-RD}
	Every Type 2 Cauchy-based MDS matrix of odd order admits a related differential pair.
\end{corollary}

In~\cite{Liu2016GeneralizedCirculant}, the authors introduce a class of matrices called left-circulant matrices. An $n\times n$ matrix $M$ is said to be a left-circulant matrix if each successive row is obtained by a left shift of the preceding row, i.e.
\[
M=l\text{-}Circ(x_1,x_2,\ldots,x_{n})
=\begin{pmatrix}
	x_1 & x_2 & \ldots & x_{n}\\
	x_{2} & x_3 & \ldots & x_{1}\\
	\vdots & \vdots &\ddots &\vdots \\
	x_n & x_1 & \ldots & x_{n-1}\\
\end{pmatrix}.
\]
It is important to note that a left-circulant matrix is symmetric. Theorem~\ref{thm:symmetric} therefore yields the following corollary.

\begin{corollary}\label{Coro:left-circulant-RD}
	Every left-circulant MDS matrix of odd order admits a related differential pair.
\end{corollary}

\begin{remark}
	It is important to note that even though the MDS property is necessary to avoid related differentials, it is not sufficient, as shown in Corollaries~\ref{Coro:Type2Cauchy-RD} and~\ref{Coro:left-circulant-RD}, where every Type 2 Cauchy-based MDS matrix of odd order and every left‑circulant MDS matrix of odd order admit a related differential pair.
\end{remark}

\noindent In the next section, we prove that circulant matrices of certain orders contain related differentials.

\section{Related Differentials and Circulant Matrices}\label{Sec:RDvsCirculant}
Circulant matrices represent one of the most widely deployed construction classes for MDS diffusion layers in cryptography. Their algebraic structure enables efficient implementation through polynomial arithmetic, while their regularity simplifies hardware design~\cite{Li2016CirculantInvolutory,Liu2016GeneralizedCirculant}. However, this same structural regularity creates systematic vulnerabilities to related-differential attacks. This section provides a broad characterization of when circulant MDS matrices possess related differentials, developing three fundamental lemmas based on modular conditions on the matrix dimension and culminating in a unified theorem (Theorem~\ref{thm:circulant}) that identifies a broad class of vulnerable orders.

A circulant matrix $M$ of order $n$ over $\mathbb{F}_{2^{m}}$ is completely determined by its first row vector $[x_1,x_2,\ldots,x_n]$, with each subsequent row obtained by cyclically shifting the previous row one position to the right.
\begin{definition}
	An $n\times n$ matrix $M$ is said to be a circulant matrix if its elements are determined by the elements of its first row $x_1,x_2,\ldots,x_n$ as
	\[
	M=Circ(x_1,x_2,\ldots,x_{n})
	=\begin{pmatrix}
		x_1 & x_2 & \ldots & x_{n}\\
		x_{n} & x_1 & \ldots & x_{n-1}\\
		\vdots & \vdots &\vdots &\vdots \\
		x_2 & x_3 & \ldots & x_1\\
	\end{pmatrix}.
	\]
\end{definition}

To align this matrix structure with polynomial arithmetic, it is necessary to examine its column vectors. The first column of $M$ is exactly $[x_1, x_n, \ldots, x_2]^T$. Let us define this first column as the coefficient vector $a = [a_0, a_1, \dots, a_{n-1}]^T$. In this column oriented view, each subsequent column is obtained by a cyclic downward shift of the previous one. For an input vector $b = [b_0, b_1, \dots, b_{n-1}]^T$, the matrix product $c = Mb$ computes the cyclic convolution of the coefficient vectors $a$ and $b$. Equivalently, the cyclic convolution of length $n$ vectors over $\mathbb{F}_{2^{m}}$ is isomorphic to polynomial multiplication in the quotient ring $R$~\cite{FJ77}, where
\begin{equation*}
	R = \mathbb{F}_{2^{m}}[X]/\langle X^n - 1 \rangle.
\end{equation*}
Under this correspondence, the vectors $a$ and $b$ are identified with the polynomials
\begin{equation*}
	M(X) = \sum_{i=0}^{n-1} a_i X^i \qquad \text{and} \qquad B(X) = \sum_{i=0}^{n-1} b_i X^i,
\end{equation*}
and the output vector $c$ corresponds to the polynomial $C(X) = \sum_{i=0}^{n-1} c_i X^i$ satisfying
\begin{equation*}
	C(X) \equiv M(X)B(X) \pmod{X^n - 1}.
\end{equation*}

% column vector $[a_{0}, a_{1}, \dots, a_{n-1}]^{T}$. Each subsequent column is obtained by a cyclic downward shift of the previous one. For a vector $b = [b_{0}, b_{1}, \dots, b_{n-1}]^{T}$, the product $c=Mb$ computes the cyclic convolution of the coefficient vectors a and b. Equivalently, cyclic convolution of length n vectors over $\mathbb{F}_{2^{m}}$ is isomorphic to polynomial multiplication in the quotient ring
% \begin{equation*}
	% R=\mathbb{F}_{2^{m}}[X]/\langle X^{n}-1\rangle.
	% \end{equation*}

% Under this correspondence, the vectors a and b are identified with the polynomials
% \begin{equation*}
	% M(X)=\sum_{i=0}^{n-1}a_{i}X^{i}
	% \end{equation*}
% and
% \begin{equation*}
	% B(X)=\sum_{i=0}^{n-1}b_{i}X^{i},
	% \end{equation*}
% and the output vector c corresponds to the polynomial satisfying
% \begin{equation*}
	% C(X)=\sum_{i=0}^{n-1}c_{i}X^{i}
	% \end{equation*}
% \begin{equation*}
	% C(X)\equiv M(X)B(X)~(\text{mod } X^{n}-1).
	% \end{equation*}

This polynomial framework enables a systematic analysis of related differentials by decomposing $M(X)$ according to the modular structure of its exponents modulo $n$.

\begin{lemma}\label{lemma:4n}
	Let $M$ be a circulant MDS matrix of order $n$ over $\mathbb{F}_{2^{m}}$. If $n$ is divisible by $4$, that is, $n \equiv 0 \pmod{4}$, then $M$ has related differentials.
\end{lemma}
\begin{proof}
	% \textcolor{red}{Let $M$ be a circulant matrix of order $n \times n$ over $\mathbb{F}_{2^{m}}$. If $M$ is non-MDS matrix then using Theorem \ref{thm:non-mds}, we get a related differential pair. Therefore, let us assume that $M$ is a MDS circulant matrix.}\\
    
	Let $M(X)\in R$ be the polynomial representation of the circulant MDS matrix $M$ and $n=4t$. We partition $M(X)$ into two distinct polynomials containing its even degree and odd degree terms, respectively:
	\begin{align*}
		M_{even}(X)&=a_{0}+a_{2}X^{2}+a_{4}X^{4}+\dots+a_{n-2}X^{n-2}= \sum_{k=0}^{2t-1}a_{2k}X^{2k} \\
		M_{odd}(X)&=a_{1}X+a_{3}X^{3}+a_{5}X^{5}+\dots+a_{n-1}X^{n-1} =\sum_{k=0}^{2t-1}a_{2k+1}X^{2k+1}
	\end{align*}
	such that $M(X)=M_{even}(X)+M_{odd}(X)$.
	
	We define two input difference polynomials, $B(X)$ and $B^{\prime}(X)$ directly from the components of the transformation matrix:
	\begin{equation*}
		B(X)=M_{even}(X) \qquad \text{and} \qquad B^{\prime}(X)=M_{odd}(X).
	\end{equation*}
	% \begin{equation*}
		% B^{\prime}(X)=M_{odd}(X).
		% \end{equation*}
	
	By definition, $B(X)$ contains nonzero coefficients exclusively at even indices, and $B^{\prime}(X)$ contains nonzero coefficients exclusively at odd indices. Consequently, they occupy perfectly disjoint coordinate spaces. For any index $j\in\{0,1,\dots,n-1\}$, the product of their corresponding coefficients evaluates to exactly zero $(b_{j}b_{j}^{\prime}=0)$. This trivially satisfies the formal algebraic condition for related differences $b_{j}b_{j}^{\prime}(b_{j}+b_{j}^{\prime})=0$.
	
	Let $C(X)$ and $C^{\prime}(X)$ represent the output difference polynomials resulting from the linear transformation. Using the ring isomorphism:
	\begin{align*}
		C(X) &= B(X)M(X) ~(\text{mod } X^{n}-1)= M_{even}(X)(M_{even}(X)+M_{odd}(X)) \\
		C^{\prime}(X) 
        &= B^{\prime}(X)M(X) ~(\text{mod } X^{n}-1)= M_{odd}(X)(M_{even}(X)+M_{odd}(X)).
	\end{align*}
	
	Therefore, we have
	\begin{align*}
		C(X) &= M_{even}(X)^{2}+M_{even}(X)M_{odd}(X) ~(\text{mod } X^{n}-1) \\
		C^{\prime}(X) &= M_{odd}(X)^{2}+M_{even}(X)M_{odd}(X) ~(\text{mod } X^{n}-1).
	\end{align*}
	
	Note that
	% \begin{equation*}
		% A_{even}(X)^{2}=(a_{0}+a_{2}X^{2}+a_{4}X^{4}+\dots)^{2}=a_{0}^{2}+a_{2}^{2}X^{4}+a_{4}^{2}X^{8}+\dots=\sum a_{2k}^{2}X^{4k} \pmod{X^{n}-1}
		% \end{equation*}
	
	\begin{equation*}
		\begin{aligned}
			M_{even}(X)^{2}&=(a_{0}+a_{2}X^{2}+a_{4}X^{4}+\dots + a_{2(2t-1)}X^{2(2t-1)})^{2}\\ &=a_{0}^{2}+a_{2}^{2}X^{4}+a_{4}^{2}X^{8}+\dots+a_{2(2t-1)}X^{4(2t-1)}\\
			&=\sum_{k=0}^{2t-1}a_{2k}^{2}X^{4k} ~(\text{mod } X^{n}-1).
		\end{aligned}
	\end{equation*}
	
	Thus, $M_{even}(X)^{2}$ strictly generates terms with exponents $j\equiv0 \pmod{4}$.
	
	Similarly,
	\begin{equation*}
		\begin{aligned}
			M_{odd}(X)^{2}&=(a_{1}X+a_{3}X^{3}+\dots + a_{4t-1}X^{4t-1})^{2}=a_{1}^{2}X^{2}+a_{3}^{2}X^{6}+\dots+a_{4t-1}X^{2(4t-1)}\\
			&=\sum_{k=0}^{2t-1}a_{2k+1}^{2}X^{4k+2} ~(\text{mod } X^{n}-1).
		\end{aligned}
	\end{equation*}
	
	Thus, $M_{odd}(X)^{2}$ strictly generates terms with exponents $j\equiv2 \pmod{4}$. It produces no constant term and no terms where $j \not\equiv 2 \pmod{4}$.
	
	$M_{even}(X)M_{odd}(X)$ represents the product of even and odd exponents. The sum of an even integer and an odd integer is always odd. Thus, it strictly generates terms with odd exponents.
	
	Also, note that even after reduction $\pmod{X^{n}-1}$, the indices will still remain in $j\equiv0 \pmod{4}$ for $M_{even}(X)^{2}$ and $j\equiv2 \pmod{4}$ for $M_{odd}(X)^{2}$ as $n$ is a multiple of $4$ ($n=4t$). The same for $M_{even}(X)M_{odd}(X)$.
	\vspace{1em}
	
	\noindent \textbf{Case A (Odd indices j):} For any odd index j, $M_{even}(X)^{2}$ and $M_{odd}(X)^{2}$ yields no terms. The coefficients $c_{j}$ and $c_{j}^{\prime}$ are derived entirely and exclusively from the shared term $M_{even}(X)M_{odd}(X)$. Therefore, $c_{j}=c_{j}^{\prime}$. In characteristic 2, this implies $c_{j}+c_{j}^{\prime}=0$.
	\vspace{1em}
	
	\noindent \textbf{Case B (Even indices $j\equiv0 \pmod{4}$):} For these indices, $M_{odd}(X)^{2}$ yields no terms, and the product $M_{even}(X)M_{odd}(X)$ yields no terms here. Therefore, the corresponding coefficient in the related output $C^{\prime}(X)$ evaluates to exactly zero: $c_{j}^{\prime}=0$.
	\vspace{1em}
	
	\noindent \textbf{Case C (Even indices $j\equiv2 \pmod{4}$):} For these indices, $M_{even}(X)^{2}$ yields no terms, and the product $M_{even}(X)M_{odd}(X)$ yields no terms here. Therefore, the corresponding coefficient in the primary output $C(X)$ evaluates to exactly zero: $c_{j}=0.$
	
	Combining all three cases, we get $c_jc_j^\prime(c_j+c_j^\prime)=0 \ \forall \ 0 \leq j \leq n-1$.
	Thus, $(B,C)$ and $(B^\prime, C^\prime)$ form related differentials.
\end{proof}

\noindent The second structural vulnerability arises when the matrix order is divisible by 3, enabling a three way polynomial decomposition. 

\begin{lemma}\label{lemma:3n}
	Let $M$ be a circulant MDS matrix of order $n$ over $\mathbb{F}_{2^{m}}$. If $n$ is divisible by $3$, that is, $n \equiv 0 \pmod{3}$, then $M$ has related differentials.
\end{lemma}
\begin{proof}
	% \textcolor{red}{Let $M$ be a circulant matrix of order $n \times n$ over $\mathbb{F}_{2^{m}}$. Let us assume that $M$ is a MDS circulant matrix, as non-MDS case can be handled by Theorem \ref{thm:non-mds}.}
    
	Let $M(X)\in R$ be the polynomial representation of the circulant MDS matrix $M$ and $n=3t$. We partition $M(X)$ into three distinct polynomials $M_j$ containing its $j \pmod3$ degree terms, respectively for $0\leq j \leq2$:
	\begin{align*}
		M_{0}(X)&=a_{0}+a_{3}X^{3}+a_{6}X^{6}+\dots+a_{n-3}X^{n-3}=\sum_{k=0}^{t-1}a_{3k}X^{3k} \\
		M_{1}(X)&=a_{1}X^{1}+a_{4}X^{4}+a_{7}X^{7}+\dots+a_{n-2}X^{n-2}=\sum_{k=0}^{t-1}a_{3k+1}X^{3k+1}\\
		M_{2}(X)&=a_{2}X^{2}+a_{5}X^{5}+a_{8}X^{8}+\dots+a_{n-1}X^{n-1}=\sum_{k=0}^{t-1}a_{3k+2}X^{3k+2}
	\end{align*}
	such that $M(X)=M_{0}(X)+M_{1}(X)+M_2(X)$.
	
	We define two input difference polynomials, $B(X)$ and $B^{\prime}(X)$ directly from the components of the transformation matrix:
	\begin{equation*}
		B(X)=M_{0}(X)+M_1(X)
	\end{equation*}
	\begin{equation*}
		B^{\prime}(X)=M_{0}(X)+M_2(X).
	\end{equation*}
	This implies 
	\begin{equation*}
		B(X)+B^{\prime}(X)=M_{1}(X)+M_2(X).
	\end{equation*}
	
	By definition, $B(X)$ contains nonzero coefficients exclusively at indices congruent to $0$ or $1 \pmod 3$, and $B^{\prime}(X)$ contains nonzero coefficients exclusively at indices congruent to $0$ or $2 \pmod 3$. Similarly, $B(X)+B^{\prime}(X)$ contains nonzero coefficients exclusively at indices congruent to $1$ or $2 \pmod 3$.
	Consequently, for any index $j\in\{0,1,\dots,n-1\}$, the product of their corresponding coefficients evaluates to exactly zero  $b_{j}b_{j}^{\prime}(b_{j}+b_{j}^{\prime})=0$. This trivially satisfies the formal algebraic condition for related differences.
	
	Let $C(X)$ and $C^{\prime}(X)$ represent the output difference polynomials resulting from the linear transformation. Using the ring isomorphism:
	\begin{align*}
		C(X)& = B(X)M(X)\\
		&= (M_{0}(X)+M_1(X))(M_{0}(X)+M_{1}(X) +M_2(X)) ~(\text{mod } X^{n}-1)  \\
		C^{\prime}(X) &= B^{\prime}(X)M(X) \\
		&= (M_{0}(X)+M_2(X))(M_{0}(X)+M_{1}(X)+M_2(X)) ~(\text{mod } X^{n}-1).
	\end{align*}
	
	Therefore, we have
	\begin{align*}
		C(X) &= M_{0}(X)^{2}+M_{1}(X)^{2}+M_{0}(X)M_{2}(X)+M_1(X)M_2(X) ~(\text{mod } X^{n}-1) \\
		C^{\prime}(X) &= M_{0}(X)^{2}+M_{2}(X)^{2}+M_{0}(X)M_{1}(X)+M_1(X)M_2(X) ~(\text{mod } X^{n}-1). 
	\end{align*}
    
	We also have 
    
	\begin{align*}
	    C(X)+C^{\prime}(X) &= M_{1}(X)^{2}+M_{2}(X)^{2}+M_{0}(X)M_{1}(X)+M_0(X)M_2(X) ~(\text{mod } X^{n}-1).
    \end{align*}

	Note that
   \begin{align*}
    M_{0}(X)^{2}&=(a_{0}+a_{3}X^{3}+a_{6}X^{6}+\dots)^{2}= \sum_{k=0}^{t-1}a_{3k}^{2}X^{6k} \pmod{X^{n}-1}
    \end{align*}
    and
    \[
    \begin{aligned}
    M_1(X)M_2(X) &= \left(\sum_{k=0}^{t-1}a_{3k+1}X^{3k+1}\right)\left(\sum_{k=0}^{t-1}a_{3k+2}X^{3k+2}\right) \\
	&= \sum_{j=0}^{t-1}\sum_{i=0}^{t-1}a_{3j+1}a_{3i+2}X^{3(i+j+1)}.
    \end{aligned}
    \]
	% \begin{equation*}
	% 	\begin{aligned}
	% 		M_{0}(X)^{2}&=(a_{0}+a_{3}X^{3}+a_{6}X^{6}+\dots)^{2} \\
	% 		&=\sum_{k=0}^{2t-1}a_{3k}^{2}X^{6k} \pmod{X^{n}-1},\\
	% 		 \text{and} & \\
	% 		M_1(X)M_2(X) &= \left(\sum_{k=0}^{t-1}a_{3k+1}X^{3k+1}\right)\left(\sum_{k=0}^{t-1}a_{3k+2}X^{3k+2}\right) \\
	% 		&= \sum_{j=0}^{t-1}\sum_{i=0}^{t-1}a_{3j+1}a_{3i+2}X^{3(i+j+1)}
	% 	\end{aligned}
	% \end{equation*}
	
	Thus, $M_{0}(X)^{2}$ and $M_1(X)M_2(X)$ strictly generate terms with exponents $j\equiv0 \pmod{3}$.
	
	Similarly,
    \[
    M_{1}(X)^{2}=\left(\sum_{k=0}^{t-1}a_{3k+1}X^{3k+1}\right)^2 
	= \sum_{k=0}^{t-1}a_{3k+1}^2X^{6k+2}
    \]
    and
	\begin{equation*}
		\begin{aligned}
			% M_{1}(X)^{2}&=\left(\sum_{k=0}^{t-1}a_{3k+1}X^{3k+1}\right)^2 
			% = \sum_{k=0}^{t-1}a_{3k+1}^2X^{6k+2},& \\
			% \text{and} \\
			M_0(X)M_2(X) &= \left(\sum_{k=0}^{t-1}a_{3k}X^{3k}\right)\left(\sum_{k=0}^{t-1}a_{3k+2}X^{3k+2}\right) \\
			&= \sum_{j=0}^{t-1}\sum_{i=0}^{t-1}a_{3j}a_{3i+2}X^{3(i+j)+2}.
		\end{aligned}
	\end{equation*}
	
	Thus, $M_{1}(X)^{2}$ and $M_0(X)M_2(X)$ strictly generate terms with exponents $j\equiv 2 \pmod{3}$. It produces no terms where $j \not\equiv 2 \pmod{3}$.
	
	Similarly,
    \[
    M_{2}(X)^{2}=\left(\sum_{k=0}^{t-1}a_{3k+2}X^{3k+2}\right)^2 
	= \sum_{k=0}^{t-1}a_{3k+2}^2X^{6k+4}
    \]
    and
	\begin{equation*}
		\begin{aligned}
			% M_{2}(X)^{2}&=\left(\sum_{k=0}^{t-1}a_{3k+2}X^{3k+2}\right)^2 
			% = \sum_{k=0}^{t-1}a_{3k+2}^2X^{6k+4},& \\
			% \text{and} \\
			M_0(X)M_1(X) &= \left(\sum_{k=0}^{t-1}a_{3k}X^{3k}\right)\left(\sum_{k=0}^{t-1}a_{3k+1}X^{3k+1}\right) \\
			&= \sum_{j=0}^{t-1}\sum_{i=0}^{t-1}a_{3j}a_{3i+1}X^{3(i+j)+1}.
		\end{aligned}
	\end{equation*}
	
	Thus, $M_{2}(X)^{2}$ and $M_0(X)M_1(X)$ strictly generate terms with exponents $j\equiv 1 \pmod{3}$. It produces no terms where $j \not\equiv 1 \pmod{3}$.
	
	% \vspace{1em}
	
	Also note that even after reduction $\pmod{X^{n}-1}$, any term in $M_i(X)M_j(X)$ will still remain in $M_i(X)M_j(X)$, as $n$ is multiple of 3 ($n=3t$).
	\vspace{1em}

	\noindent \textbf{Case A ($j \equiv 0 \pmod 3$):} For any index $j$ with $j \equiv 0 \pmod 3$, only $M_0(X)^2$ and $M_1(X)M_2(X)$ polynomial contributes terms. The coefficients $c_{j}$ and $c_{j}^{\prime}$ are derived entirely and exclusively from the shared terms $M_0(X)^2$ and $M_{1}(X)M_{2}(X)$. Therefore, $c_{j}=c_{j}^{\prime}$. In characteristic 2, this implies $c_{j}+c_{j}^{\prime}=0$.
	\vspace{1em}
	
	\noindent \textbf{Case B ($j\equiv1 \pmod{3}$):} For these indices, only  $M_{2}(X)^{2}$ and $M_0(X)M_1(X)$ contribute terms. $C(X)$ does not contain any of these terms, therefore, the corresponding coefficient in the output $C(X)$ evaluates to exactly zero: $c_{j}=0$.
	\vspace{1em}
	
	\noindent \textbf{Case C ($j\equiv2 \pmod{3}$):} For these indices only $M_{1}(X)^{2}$ and $M_0(X)M_2(X)$ contribute terms. Since $C^\prime(X)$ does not contain any of these terms, the corresponding coefficient in the related output $C^\prime(X)$ evaluates to exactly zero: $c'_{j}=0.$
	
	Combining all three cases, we get $c_jc_j^\prime(c_j+c_j^\prime)=0 \ \forall \ 0 \leq j \leq n-1$.
	Thus, $(B,C)$ and $(B^\prime, C^\prime)$ form related differentials.
\end{proof}

The third structural result exploits the symmetry property of odd order left-circulant matrices (see Corollary~\ref{Coro:left-circulant-RD}) and relies on the following proposition.

\begin{proposition}\label{Prop:permutationRD}
	Let $M$ be an $n \times n$ matrix. Then the existence of related differentials is invariant under the transformation $PMQ$, where $P$ and $Q$ are permutation matrices.
\end{proposition}
\begin{proof}
	Let $P$ and $Q$ be permutation matrices, and let $u=[u_1,u_2,\dots,u_n]^T$ and $v=[v_1,v_2,\dots,v_n]^T$ be related differences such that $(u,Mu)$ and $(v,Mv)$ form a related differential pair for $M$. Then 
	\[
	u_i \cdot v_i \cdot (u_i \oplus v_i) = 0 \text{ and }  (Mu)_i \cdot (Mv)_i \cdot ((Mu)_i \oplus (Mv)_i) = 0 
	\]
    $\forall ~i \in \{1,2,\ldots,n\}$. Define new vectors
	\[
	u' = Q^{-1}u \qquad \text{and} \qquad v' = Q^{-1}v.
	\]
	Since $Q^{-1}$ is a permutation matrix, $u'$ and $v'$ are obtained from $u$ and $v$ by permuting their coordinates. Hence, for every $i$, we still have
	\[
	u'_i \cdot v'_i \cdot (u'_i \oplus v'_i)=0.
	\]
	
	Now
	\[
	(PMQ)u' = PMu \qquad \text{and} \qquad (PMQ)v' = PMv.
	\]
	Since $P$ is also a permutation matrix, the vectors $PMu$ and $PMv$ are obtained from $Mu$ and $Mv$ by permuting their coordinates. Therefore, for every $i$,
	\[
    \begin{split}
    &(PMQu')_i \cdot (PMQv')_i \cdot \big((PMQu')_i \oplus (PMQv')_i\big) = 0.
    \end{split}
	% \big((PMQ)u'\big)_i \cdot \big((PMQ)v'\big)_i \cdot
	% \left(\big((PMQ)u'\big)_i \oplus \big((PMQ)v'\big)_i\right)=0.
	\]
	Thus, $u'$ and $v'$ are related differences for $PMQ$. The converse follows by applying the same argument to $M = P^{-1}(PMQ)Q^{-1}$. This completes the proof.
	% Now  $((PMQ)u)_i=(Mu)_i \quad \forall ~i \in \{1,2\dots,n\}$. Then, we have 
	%\begin{align*}
	%    &((D_1M_1D_2)u)_i \cdot ((D_1M_1D_2)v)_i \cdot (((D_1M_1D_2)u)_i \oplus ((D_1M_1D_2)v)_i)\\
	%&= \alpha_i^3\beta_i^3(M_1u)_i \cdot (M_1v)_i \cdot ((M_1u)_i \oplus (M_1v)_i) = 0 
	%\quad \forall ~i \in \{1,2\dots,n\}.
	%\end{align*}
	% 
	% 
	% Similarly, we can show that if $u$ and $v$ are not related differences, then $(D_1M_1D_2)u$ and $(D_1M_1D_2)v$ are also not related differences. 
	% 
\end{proof}

It is worth noting that a left-circulant matrix is a row-permuted circulant matrix. More specifically, for $M=Circ(x_1,x_2,\ldots,x_n)$,  we have  $PM=l\text{-}Circ(x_1,x_2,\ldots,x_n)$, where $P$ is a permutation matrix given by   
\[
P=
\begin{pmatrix}
	1&0&0&\ldots&0&0&0\\ 
	0&0&0&\ldots&0&0&1\\
	0&0&0&\ldots&0&1&0\\
	\vdots&\vdots&\vdots&\ldots&\vdots&\vdots&\vdots\\
	0&1&0&\ldots&0&0&0\\
\end{pmatrix}.
\]
We now use this fact to prove the following result on circulant matrices.

% \begin{remark} \cite[Ramark 30]{kcz} \label{l-circ_remark}
	%  It is worth mentioning that a left-circulant matrix is a row-permuted circulant matrix. More specifically, for $M=Circ(x_1,x_2,\ldots,x_n)$,  we have  $PM=l\text{-}Circ(x_1,x_2,\ldots,x_n)$, where $P$ is permutation matrix given by   
	%  \[
	%  P=
	%  \begin{bmatrix}
		%     1&0&0&\ldots&0&0&0\\ 
		%     0&0&0&\ldots&0&0&1\\
		%     0&0&0&\ldots&0&1&0\\
		%     \vdots&\vdots&\vdots&\ldots&\vdots&\vdots&\vdots\\
		%      0&1&0&\ldots&0&0&0\\
		%  \end{bmatrix}
	%  .\]
	% \end{remark}

\begin{lemma}\label{lemma:odd}
	Let $M$ be a circulant MDS matrix of order $n$ over $\mathbb{F}_{2^{m}}$. If $n$ is odd, that is, $n \equiv 1 \pmod2$, then $M$ has related differentials.
\end{lemma}
\begin{proof}
	% \textcolor{red}{Let $M$ be a circulant matrix of odd order over $\mathbb{F}_{2^{m}}$. In Theorem \ref{thm:non-mds} we have already shown the existence of a nontrivial related differential pair for a non-MDS matrix. Thus, we assume that $M$ is a MDS matrix.}
    Let $M$ be a circulant MDS matrix of odd order $n$. 
    By the above observation, there exists a left-circulant matrix $M'$ and a permutation matrix $P$ such that $M=PM'$. From Proposition~\ref{Prop:permutationRD}, we know that related-differential property is invariant under left multiplication by a permutation matrix. Moreover, a left-circulant matrix is always symmetric, so by Corollary~\ref{Coro:left-circulant-RD}, $M'$ has a pair of related differentials. It follows that $M$ also admits a related differential pair.
\end{proof}

% Combining Lemma \ref{lemma:4n}, Lemma \ref{lemma:3n} and Lemma \ref{lemma:odd}, we get the following theorem.

We now combine Lemmas \ref{lemma:4n}, \ref{lemma:3n}, and \ref{lemma:odd} to obtain a characterization of circulant MDS matrices with related differentials.

\begin{theorem}\label{thm:circulant}
	Let $M$ be a circulant MDS matrix of order $n$ over $\mathbb{F}_{2^{m}}$. If $n \not\equiv \pm 2 \pmod{12}$, then $M$ has related differentials.
\end{theorem}
% \medskip
% \textcolor{red}{can we conjecturise the converse of above theorem as well. Example using sage code is needed to verify validity.\\
	% -- I tried with the Sagemath code (that I gave you) for order $n=10$ over $GF(2^8)$-- each circulant MDS matrix has RD, but it was taking long time over $GF(2^{10})$ and $GF(2^{16})$. To me it seems hard to find a counter example as this result may be true even for those orders, but cant say now! So better to leave as conjecture now. If possible you could try to search for restricted domain in the FF to have a counter example.-- Susanta}

\noindent In the next section, we provide the characterization of MDS matrices that are free from related differentials. In particular, we present a generalized analysis of $3\times3$ MDS matrices by examining the existence of vectors 
that constitute related differentials.

\section{Characterization of Related Differentials over $3\times 3$ MDS Matrices}\label{Sec:RD_Charac}

To facilitate this analysis, we first establish Lemmas~\ref{lemma_bound} and \ref{lemma 1}, which exploit the MDS property to characterize all valid pairs of related input differences having minimum Hamming weights. These pairs serve as the fundamental building blocks for deriving the conditions under which the considered mapping admits related differentials. 

We now restate below the following lemma from~\cite{daemen2009new}, which bounds the weight of the input/output difference pairs forming related differentials.

\begin{lemma}\label{lemma_bound}
	If $(u,v)$ and $(u',v')$ are related differentials over a linear map with an associated $n\times n$ matrix that is MDS, then  
	\[ 
    \begin{aligned}
    &\min\bigl\{
      \operatorname{wt}(u)+\operatorname{wt}(v),
      \operatorname{wt}(u')+\operatorname{wt}(v'), \operatorname{wt}(u \oplus u')+\operatorname{wt}(v\oplus v')
      \bigr\}
      \leq n+\left\lfloor\frac{n}{3}\right\rfloor.
    \end{aligned}
    % \min\left\{\operatorname{wt}(u) + \operatorname{wt}(v), \operatorname{wt}(u') + \operatorname{wt}(v'), \operatorname{wt}(u+u') + \operatorname{wt}(v+v')\right\}\leq n+\left\lfloor\frac{n}{3}\right\rfloor.
    \]
\end{lemma}

This observation implies that, to determine whether a matrix $M$ admits related differentials, it is sufficient to examine all input/output difference pairs whose total (combined) Hamming weight is at most 
\(
n + \left\lfloor \frac{n}{3} \right\rfloor.
\) 
Indeed, analyzing pairs up to this weight threshold is enough to detect every possible related differential pair associated with the matrix.

Let us define a vector $t= u \oplus v$ where $u$ and $v$ are related differences. Then observe that any pair of vectors of the triplet $(u,v,t)$ are related differences. So we call $(u,v,t)$ a related difference triplet or ``related triplet'' for short. 
%Similar to trivial related differences as mentioned in Remark \ref{remark-trivial}, $(u,u,0)$ always forms a related triplet for any $u \in \mathbb{F}_{2^m}^n$. We call this triplet a trivial related triplet. Since we are trying to determine nontrivial related differences, we omit this trivial triplet.

Now observe that for a linear map $M$, if $(u,v)$ and $(Mu,Mv)$ are related differences, then any two of the three differentials $(u, Mu)$, $(v, Mv)$, $(t, Mt)$ form related differentials, where $t=u\oplus v$.

In fact, any permutation of related triplets $(u,v,t)$ and $(Mu, Mv, Mt)$ preserves the related-differential property among them. This forms an equivalence class of such triplets. Moreover, to study the related-differential property of a related triplet $(u,v, u\oplus v)$ over a linear map $M$, where $u$ and $v$ are related differences, it is enough to determine conditions for any one pair of this triplet. 

We now restate the following lemma from~\cite{jha2025construction}.

\begin{lemma}\label{lemma 1}
	If $(u,v,t)$ are related differences triplet in $\mathbb{F}_{2^m}^n$ with $t=u\oplus v$, then
	\[\operatorname{wt}(u) + \operatorname{wt}(v) + \operatorname{wt}(t ) \leq 2n.\]
\end{lemma}

%The goal is to determine the precise conditions under which a given matrix allows related differentials when two of the input differences of the triplet $(u,v,t)$ are of minimal weight. In our analysis, we consider the case when n=3, From Lemma~\ref{lemma 1}, we know that $\operatorname{wt}(u) + \operatorname{wt}(v) + \operatorname{wt}(t ) \leq 6$. Accordingly, the analysis can proceed by systematically investigating each of these minimal weight cases.

The goal is to determine the precise conditions under which a given matrix admits related differentials when two of the input differences in the triplet $(u,v,t)$ have minimal weight.  We focus our analysis on the case $n=3$. By Lemma~\ref{lemma 1}, we have $\operatorname{wt}(u)+\operatorname{wt}(v)+\operatorname{wt}(t)\leq 6$. Accordingly, we systematically investigate the possible weight distributions of the triplet $(u,v,t)$ that satisfy this bound, with particular emphasis on the cases in which two input differences have minimal weight.

\begin{lemma} \label{lemma_poss_triplet}
Let $M$ be a $3\times 3$ matrix over $\mathbb{F}_{2^m}$. Then, any related triplet $(u,v,t)$, where $t= u \oplus v$ and $(u,v)$ are related differences, has the following possible weight cases: 
\[(1,1,2),(1,2,3) \text{ and } (2,2,2).\]
\end{lemma}
\begin{proof}
Let us investigate related differences based on the possible weights of vectors. Let $u=[u_1,u_2,u_3]^T$,  and $v=[v_1,v_2,v_3]^T$ be related differences and $t=u\oplus v =[u_1\oplus v_1, u_2 \oplus v_2, u_3 \oplus v_3 ]^T$. We have the following cases:

\textbf{Case 1.} Let $\operatorname{wt}(u)=1$ and $\operatorname{wt}(v)=1$. Since we are considering only nontrivial related differences, $u$ and $v$ are distinct, and their unique nonzero coordinates occur at different indices. Thus, we have $\operatorname{wt}(t)=2$. Therefore, we get a related triplet $(u,v,t)$ with weights $(1,1,2)$. 

\textbf{Case 2.} Suppose that $\operatorname{wt}(u)=1$ and $\operatorname{wt}(v)=2$. Now, we have two possibilities. 
First, one nonzero coordinate of $u$ and $v$ lies at the same index, say $u_i=v_i$. Then we have $\operatorname{wt}(t)=1$. This related triplet $(u,v,t)$ with weights $(1,2,1)$ is equivalent to the triplet of Case 1.\\
Secondly, if the non-zero coordinates of $u$ and $v$ are at disjoint indices. Then we have  $\operatorname{wt}(t)=3$. Thus, we get a related triplet $(u,v,t)$ with weights $(1,2,3)$. 

\textbf{Case 3.}  Suppose that $\operatorname{wt}(u)=1$ and $\operatorname{wt}(v)=3$. Here, the nonzero coordinate of  $u$ and $v$ at a certain index is the same.  Then, $\operatorname{wt}(t)=2$. Thus, the related triplet $(u,v,t)$ with weights $(1,3,2)$ is equivalent, up to permutation, to the triplet in Case 2.

\textbf{Case 4.} Suppose that $\operatorname{wt}(u)=2$. If $\operatorname{wt}(v)=1$, then this is equivalent to Case 2. Now, suppose that $\operatorname{wt}(v)=2$. For related differences, the nonzero coordinate of  $u$ and $v$ at a certain index is the same. Thus, we get a related triplet $(u,v,t)$ with weights $(2,2,2)$.

\textbf{Case 5.} Suppose that $\operatorname{wt}(u)=2$ and $\operatorname{wt}(v)=3$. Here, at certain indices, two nonzero coordinates of  $u$ and $v$ are equal. Thus, we get a related triplet $(u,v,t)$ with weights $(2,3,1)$ which is equivalent, up to permutation, to the triplet in Case 2.

\textbf{Case 6.} Suppose that $\operatorname{wt}(u)=3$. If $\operatorname{wt}(v)=1$ (or $\operatorname{wt}(v)=2$), then this is equivalent to  Case 3 (or Case 5), respectively. It remains to consider $\operatorname{wt}(v)=3$. However, by the observation in Section~\ref{Sec:Definition}, two distinct vectors of full weight 3 cannot form related differences. Hence, this case cannot occur.

Therefore, we have three possible related triplets of weights: $(1,1,2), (1,2,3)$ and $(2,2,2)$.
\end{proof}

\begin{lemma} \label{lemma_exclude_triplet}
If $u$ and $v$ are related differences and an input pair to a $3 \times 3$ MDS matrix
$M$ where 
\[\bigl(\operatorname{wt}(u),\operatorname{wt}(v),\operatorname{wt}(u \oplus v)\bigr)=(1,1,2),\] then $(u,Mu)$ and $(v,Mv)$ can never form related differentials.
\end{lemma}

\begin{proof}

Since $M$ is a $3 \times 3$ MDS matrix, its branch number is $4$. Hence, for any nonzero input difference $x$,
\[
\operatorname{wt}(x)+\operatorname{wt}(Mx)\geq 4.
\]

Moreover, when $\operatorname{wt}(x)=1$, the MDS property implies
\[
\operatorname{wt}(Mx)=3.
\]

Consider any related triplet $(u,v,t)$, where $t= u \oplus v$ and $(u,v)$ are related differences with the weights $(1,1,2)$.
%\[
%\bigl(\operatorname{wt}(u),\operatorname{wt}(v),\operatorname{wt}(t)\bigr)=(1,1,2).
%\]
Then
\[
\operatorname{wt}(Mu)=\operatorname{wt}(Mv)=3.
\]
By the observation in Section~\ref{Sec:Definition}, two distinct vectors of full weight 3 cannot form related differences. Consequently, $Mu=Mv$; the two differentials $(u, Mu)$ and $(v, Mv)$ cannot constitute a nontrivial related differential pair.

Therefore, any related triplet $(u,v,t)$ having weights $(1,1,2)$ cannot produce related differentials through a $3\times3$ MDS matrix.
\end{proof}

Omitting the invalid weights identified in Lemma~\ref{lemma_exclude_triplet}, we observe that related triplets with weights $(1,2,3)$ and $(2,2,2)$ may form related differentials. We therefore consider these two classes of related triplets in our analysis presented in Theorem~\ref{thm3*3}.

% Given a vector $u\in \mathbb{F}_{2^m}^n $, define $supp(u)=\{i|u_i\neq0 ~\forall ~ 1\leq i\leq n\}$.

%\begin{proposition}
%    Let $u, v  \in \mathbb{F}_{2^m}^n $ then 
%    \[\mathrm{wt}(u\oplus v)= \mathrm{wt}(u)+\mathrm{wt}(v)-2\cdot | \mathrm{supp}(u)\cap \mathrm{supp}(v)|. \]
%\end{proposition}

\begin{lemma}\label{lemma 2variable}
	Let $a_ix+b_iy+ c_iz=0,~i \in \{1,2,3\}$ be a system of homogeneous linear equations in three variables over $\mathbb{F}_{2^m}$. Then this system has a consistent nontrivial solution if and only if 
	\[
	\begin{vmatrix}
		a_1 & b_1 & c_1\\
		a_2 & b_2 & c_2\\
		a_3 & b_3 & c_3\\
	\end{vmatrix} = 0
	.\]
\end{lemma}

%\begin{lemma}\label{lemma 3variable}
%Let $a_ix+b_iy+c_iz=d_i,~i \in \{1,2,3,4\}$ be a system of linear equations in three variables %over $\mathbb{F}_{2^m}$. Then this system has a consistent solution if and only if 
%\[
%\begin{vmatrix}
%a_1 & b_1 & c_1& d_1\\
%a_2 & b_2 & c_2& d_2\\
%a_3 & b_3 & c_3& d_3\\
%a_3 & b_3 & c_3& d_4\\
%\end{vmatrix} = 0
%\]
%\end{lemma}

%We now find general matrices that do not admit any related differentials. In the following section, we show that any non-MDS matrix contains a related differential.

% Now we use this to determine all the necessary and sufficient conditions that the linear maps of the form $3 \times 3$ matrices will not contain a related differential. 

% \subsection{Characterization of Related Differentials over $3\times3$ MDS Matrices}
\noindent Building on those structural results, we now present a complete algebraic characterization of $3\times3$ MDS matrices that possess related differentials. This characterization provides both necessary and sufficient conditions through a systematic case by case analysis of all possible weight patterns for related triplets. We develop the theory through the representative matrix form~\cite{Kumar_MDS2024}, establish diagonal invariance, derive fifteen algebraic conditions that completely determine the related-differentials property, and provide concrete examples demonstrating both vulnerable and resistant matrices.

In~\cite{Kumar_MDS2024}, the authors introduce a technique for generating all $n \times n$ MDS and involutory MDS matrices over $\mathbb{F}_{2^m}$. The proposed method involves first identifying $n \times n$ representative MDS matrices using a search-based approach. Subsequently, all $n \times n$ MDS and involutory MDS matrices can be obtained by multiplying two diagonal matrices with these representative matrices. To find all $n \times n$ representative MDS matrices, the authors define the representative matrix of the form $M_1$ over $\mathbb{F}_{2^m}^*$ as follows:

\begin{equation}\label{Eqn_The_matrixM1}
	M_1=
	\begin{pmatrix}
		1&1&\ldots&1\\
		1& & &\\
		\vdots & &R&\\
		1&  & &
	\end{pmatrix},
\end{equation}
where $R$ is a $(n-1) \times (n-1)$ matrix.

They also present the unique decomposition of a matrix $M$ over $\mathbb{F}_{2^m}^*$ in the form $M=D_1M_1D_2$, where $D_1$ and $D_2$ are nonsingular diagonal matrices. The result is stated below.

\begin{theorem}\cite{Kumar_MDS2024}\label{Th_decomposition}
	Let $M$ be an $n \times n$ matrix over $\mathbb{F}_{2^m}^*$. Then, there exist unique $n \times n$ matrices $D_1$, $D_2$, and $M_1$ over $\mathbb{F}_{2^m}^*$ such that 
	\[M = D_1 M_1 D_2,\] 
	where $D_1$ and $D_2$ are diagonal matrices, the first entry of $D_2$ is $1$, and $M_1$ is a matrix of the form given in (\ref{Eqn_The_matrixM1}).
\end{theorem}

Multiplying a row of a matrix by a nonzero scalar is an elementary row operation, and the MDS property is preserved under such operations. For a diagonal matrix $D=\mathsf{diag}(\alpha_1,\alpha_2,\ldots,\alpha_n)$, the product $DM_1$ (respectively, $M_1D$) scales the $i$-th row (respectively, the $i$-th column) of $M_1$ by $\alpha_i$. Hence, if $M_1$ is an MDS matrix, then $D_1M_1D_2$ is also MDS for any nonsingular diagonal matrices $D_1$ and $D_2$. The following lemma shows that the related-differential property is also invariant under multiplication by nonsingular diagonal matrices.

% \noindent \textcolor{blue}{In the proof, we need to show a concrete vector pair. have a look this proof. The same for the permutation similarity}

\begin{lemma}\label{Diag_RD}
	% Let $M_1$ be a representative matrix, then related differentials are invariant under $D_1M_1 D_2$ where $D_1$ and $D_2$ are nonsingular diagonal matrices.
	For any nonsingular diagonal matrices $D_1$ and $D_2$, the matrix $M_1$ has a related differential pair if and only if $D_1M_1D_2$ has a related differential pair.
\end{lemma}

\begin{proof}
Let $D_1=\mathsf{diag}(\alpha_1,\alpha_2,\ldots,\alpha_n)$ and $D_2=\mathsf{diag}(\beta_1,\beta_2,\ldots,\beta_n)$ be nonsingular diagonal matrices, and let $u=[u_1,u_2,\dots,u_n]^T$ and $v=[v_1,v_2,\dots,v_n]^T$ be related differences such that $(u,M_1u)$ and $(v,M_1v)$
	form a related differential pair for $M_1$. Then
    \[
    \begin{aligned}
        &u_i \cdot v_i \cdot (u_i \oplus v_i) = 0 \text{ and }(M_1u)_i \cdot (M_1v)_i \cdot ((M_1u)_i \oplus (M_1v)_i) = 0
    \end{aligned}
    \]
    $\forall~i \in \{1,2,\dots,n\}$.
	% \[
	% u_i \cdot v_i \cdot (u_i \oplus v_i) = 0 \text{ and }  (M_1u)_i \cdot (M_1v)_i \cdot ((M_1u)_i \oplus (M_1v)_i) = 0 
	% \quad \forall ~i \in \{1,2,\dots,n\}.
	% \]
	To show that $D_1M_1D_2$ also has a related differential, we define new vectors $u' = D_2^{-1}u$ and $v' = D_2^{-1}v$. Since $D_2^{-1} = \mathsf{diag}(\beta_1^{-1}, \beta_2^{-1}, \dots, \beta_n^{-1})$, we have $u'_i = \beta_i^{-1}u_i$ and $v'_i = \beta_i^{-1}v_i \quad \forall ~i \in \{1,2,\dots,n\}$. Thus, we have $u'_i \cdot v'_i \cdot (u'_i \oplus v'_i)=0$.
	
	Next, we evaluate the outputs of $D_1M_1D_2$ on $u'$ and $v'$. Notice that 
	\[ (D_1M_1D_2)u' = D_1M_1u \quad \text{and} \quad (D_1M_1D_2)v' = D_1M_1v.\]
	For the $i$-th component, this yields
	\[((D_1M_1D_2)u')_i = \alpha_i(M_1u)_i, \quad ((D_1M_1D_2)v')_i = \alpha_i(M_1v)_i .\]
	Now $\forall~i \in \{1,2,\dots,n\}$, we have
	\begin{equation*}
		\begin{aligned}
			&\alpha_i(M_1u)_i \cdot \alpha_i(M_1v)_i \cdot ( \alpha_i(M_1u)_i \oplus \alpha_i(M_1v)_i)\\
			=& \alpha_i^3 \big[ (M_1u)_i \cdot (M_1v)_i \cdot ((M_1u)_i \oplus (M_1v)_i) \big] =0.
		\end{aligned}
	\end{equation*}
	Thus, the vectors $u'$ and $v'$ are related differences for  $D_1M_1D_2$. The converse follows by applying the same argument to $M_1 = D_1^{-1}(D_1M_1D_2)D_2^{-1}$. This completes the proof.
	% Now  $((D_1M_1D_2)u)_i=\alpha_i(M_1u)_i\beta_i \quad \forall ~i \in \{1,2\dots,n\}$. Then, we have 
	% \begin{align*}
		%     &((D_1M_1D_2)u)_i \cdot ((D_1M_1D_2)v)_i \cdot (((D_1M_1D_2)u)_i \oplus ((D_1M_1D_2)v)_i)\\
		% &= \alpha_i^3\beta_i^3(M_1u)_i \cdot (M_1v)_i \cdot ((M_1u)_i \oplus (M_1v)_i) = 0 
		% \quad \forall ~i \in \{1,2\dots,n\}.
		% \end{align*}
	% Thus, 
	% 
	% Similarly, we can show that if $u$ and $v$ are not related differences, then $(D_1M_1D_2)u$ and $(D_1M_1D_2)v$ are also not related differences. 
	% 
	% 
\end{proof}

% Therefore, since $D_1$ and $D_2$ are nonsingular diagonal matrices, from Lemma \ref{Diag_RD}, we can conclude that instead of checking whether the matrix $M$ has related differential or not, we need to check whether the same holds for $M_1$ or not.
Therefore, since $D_1$ and $D_2$ are nonsingular diagonal matrices, Lemma~\ref{Diag_RD} shows that $M$ has a related differential pair if and only if the representative matrix $M_1$ given in~\eqref{Eqn_The_matrixM1} has a related differential pair. Consequently, it is enough to verify the related-differential property for $M_1$.
\vspace{1em}

\begin{theorem}\label{thm3*3}
	Let
	\[ M_1= \begin{pmatrix}
		1&1&1\\
		1&a&b\\
		1&c&d\\
	\end{pmatrix}\] be a $3\times 3$ representative MDS matrix over $\mathbb{F}_{2^m}$. Then $M_1$ has no related differential pair if and only if the parameters $a,b,c,d$ do not satisfy any of the $15$ conditions $C_i(a,b,c,d)=0$, for $i=1,\dots,15$, listed in \eqref{Eqn:3MDS_RD}.
	
	\begin{equation*}
		0=\left\{
		\begin{aligned}
			a+b+c+d\\
			a+d\\
			b+c\\
		\end{aligned}
		\right.
		\quad ad+bc =\left\{ 
		\begin{aligned}
			a+b\\
			a+c\\
			b+d\\
			c+d\\ 
		\end{aligned}
		\right.
		\quad ad =\left\{ 
		\begin{aligned}
			a+b+d\\
			a+c+d\\
			b\\
			c\\
		\end{aligned}
		\right.
	\end{equation*}
    \begin{equation}\label{Eqn:3MDS_RD}
		\quad bc= \left\{ 
		\begin{aligned}
			a+b+c\\
			b+c+d\\
			a\\
			d\\
		\end{aligned}
		\right.
	\end{equation}
\end{theorem}
\begin{proof}
We already know that $(u,M_1u)$ and $(v,M_1v)$ form related differentials if and only if $(u,v,t)$ and $(M_1u,M_1v,M_1t)$ form related triplets.
From Lemmas~\ref{lemma_poss_triplet} and \ref{lemma_exclude_triplet}, for a $3\times 3$  MDS matrix, a related triplet $(u,v,t)$ must have weights \[(1,2,3) \text{ or } (2,2,2).\] First, we consider the case $\bigl(\operatorname{wt}(u),\operatorname{wt}(v),\operatorname{wt}(t)\bigr)=(1,2,3)$. 

\medskip

\noindent\textbf{Main Case I:} $\bigl(\operatorname{wt}(u),\operatorname{wt}(v),\operatorname{wt}(t)\bigr)=(1,2,3)$.

\noindent We observe that if $\bigl(\operatorname{wt}(u),\operatorname{wt}(v),\operatorname{wt}(t)\bigr)=(1,2,3)$, then using MDS property of $M_1$, we have $\operatorname{wt}(M_1u)=3$, $\operatorname{wt}(M_1v)\geq 2$ and $\operatorname{wt}(M_1t)\geq 1$. 
Thus, from Lemma~\ref{lemma 1}, $(u, M_1u)$ and $(v,M_1v)$ form related differentials if and only if $\bigl(\operatorname{wt}(M_1u),\operatorname{wt}(M_1v),\operatorname{wt}(M_1t)\bigr)=  (3,2,1) $. We have the following cases for $(u,v,t) \in \mathbb{F}_{2^m}^3$.

        \textbf{Case 1.} $u = [x, 0, 0]^T$, $v= [0, y, z]^T$, and  $t= [x, y, z]^T$
        
        \textbf{Case 2.} $u = [0, x, 0]^T$, $v= [y, 0, z]^T$, and  $t= [y, x, z]^T$
        
        \textbf{Case 3.} $u = [0, 0, x]^T$,  $v= [y, z, 0]^T$,  and   $t=  [y, z, x]^T$.
    
Now, we will discuss each case one by one.  \\
\noindent \textbf{Case 1:} $u = [x, 0, 0]^T$, $v = [0, y, z]^T$, and $t = [x, y, z]^T$.  Then we have,

	\[M_1u=[x,x,x]^T, M_1v=[y+z, ay+bz,cy+dz]^T,\]
    \[M_1t=[x+y+z, x+ay+bz,x+cy+dz]^T.\]
	
	 \noindent Now, since $ \operatorname{wt}(M_1u)=3$ and $ \operatorname{wt}(M_1t)=1$, thus, $M_1u$ and $M_1t$ will form related differences if any two coordinates of $M_1t$ will be zero, and the remaining coordinate of  $M_1t$ will be equal to the corresponding coordinate of $M_1u$, at the same index, which is $x$ in this case. Thus, we have three possible subcases.  \\

 \noindent \textbf{Subcase 1.1:} Suppose that $M_1t=[x, 0, 0]^T$, this will give us the following system of linear equations:
	\[
	M_1t= \begin{pmatrix}
		1&1&1\\
		1&a&b\\
		1&c&d\\
	\end{pmatrix}\begin{pmatrix}
		x\\
		y\\
		z\\
	\end{pmatrix} = \begin{pmatrix}
		x\\
		0\\
		0\\
	\end{pmatrix}	\]

    \[
\implies \begin{pmatrix}
		0&1&1\\
		1&a&b\\
		1&c&d\\
	\end{pmatrix}\begin{pmatrix}
		x\\
		y\\
		z\\
	\end{pmatrix}= 
    \begin{pmatrix}
		0\\
		0\\
		0\\
	\end{pmatrix}.
\]
	Applying Lemma \ref{lemma 2variable} to the above system gives us the following condition:
	\[  \begin{vmatrix}
		0&1&1\\
		1&a&b\\
		1&c&d\\
	\end{vmatrix}=0 \implies a+b+c+d=0.
	\]
	
 \noindent \textbf{Subcase 1.2:} Suppose that $M_1t=[0, x, 0]^T$. Proceeding as in Subcase~1.1 and applying Lemma~\ref{lemma 2variable}, we obtain
%     \[
% 	M_1t= \begin{pmatrix}
% 		1&1&1\\
% 		1&a&b\\
% 		1&c&d\\
% 	\end{pmatrix}\begin{pmatrix}
% 		x\\
% 		y\\
% 		z\\
% 	\end{pmatrix} = \begin{pmatrix}
% 		0\\
% 		x\\
% 		0\\
% 	\end{pmatrix}	\]	

%     \[
% \implies \begin{pmatrix}
% 		1&1&1\\
% 		0&a&b\\
% 		1&c&d\\
% 	\end{pmatrix}\begin{pmatrix}
% 		x\\
% 		y\\
% 		z\\
% 	\end{pmatrix}= 
%     \begin{pmatrix}
% 		0\\
% 		0\\
% 		0\\
% 	\end{pmatrix}
% \]
% 	Applying Lemma \ref{lemma 2variable} to the above system gives us the following condition:
	\[ \begin{vmatrix}
		1&1&1\\
		0&a&b\\
		1&c&d\\
	\end{vmatrix}=0 \implies a+b=ad+bc.
	\]

 \noindent \textbf{Subcase 1.3:} Suppose that $M_1t=[0, 0, x]^T$. Applying the same argument as in Subcase~1.1 yields
%     \[
% 	M_1t= \begin{pmatrix}
% 		1&1&1\\
% 		1&a&b\\
% 		1&c&d\\
% 	\end{pmatrix}\begin{pmatrix}
% 		x\\
% 		y\\
% 		z\\
% 	\end{pmatrix} = \begin{pmatrix}
% 		0\\
% 		0\\
% 		x\\
% 	\end{pmatrix}	\]	

%     \[
% \implies \begin{pmatrix}
% 		1&1&1\\
% 		1&a&b\\
% 		0&c&d\\
% 	\end{pmatrix}\begin{pmatrix}
% 		x\\
% 		y\\
% 		z\\
% 	\end{pmatrix}= 
%     \begin{pmatrix}
% 		0\\
% 		0\\
% 		0\\
% 	\end{pmatrix}
% \]
% 	Applying Lemma \ref{lemma 2variable} to  the above system gives us the following condition:
	\[ \begin{vmatrix}
		1&1&1\\
		1&a&b\\
		0&c&d\\
	\end{vmatrix}=0 \implies c+d=ad+bc.
	\]

 \noindent \textbf{Case 2:} $u = [0, x, 0]^T$, $v = [y, 0, z]^T$, and $t = [y, x, z]^T$.  Then we have,

	\[M_1u=[x,ax,cx]^T, M_1v=[y+z, y+bz, y+dz]^T,\] 
    \[M_1t=[x+y+z, ax+y+bz,cx+y+dz]^T.\]
	
	Similar to Case 1, again we have $3$ possible subcases. \\

  \noindent \textbf{Subcase 2.1:} Suppose that $M_1t=[x, 0, 0]^T$, this will give us the following system of linear equations:	
	
	\[
	M_1t= \begin{pmatrix}
		1&1&1\\
		1&a&b\\
		1&c&d\\
	\end{pmatrix}\begin{pmatrix}
		y\\
		x\\
		z\\
	\end{pmatrix} = \begin{pmatrix}
		x\\
		0\\
		0\\
	\end{pmatrix}
	\]
\[
\implies \begin{pmatrix}
		1&0&1\\
		1&a&b\\
		1&c&d\\
	\end{pmatrix}\begin{pmatrix}
		y\\
		x\\
		z\\
	\end{pmatrix}= 
    \begin{pmatrix}
		0\\
		0\\
		0\\
	\end{pmatrix}.
\]
	
	Again, applying Lemma \ref{lemma 2variable} to the above system gives us the following condition:
	\[ \begin{vmatrix}
		1&0&1\\
		1&a&b\\
		1&c&d\\
	\end{vmatrix}=0 \implies a+c=ad+bc.
	\]
	
 \noindent \textbf{Subcase 2.2:} Suppose that $M_1t=[0, ax, 0]^T$. Proceeding as in Subcase~1.1 and applying Lemma~\ref{lemma 2variable}, we obtain

	% \[
	% M_1t= \begin{pmatrix}
	% 	1&1&1\\
	% 	1&a&b\\
	% 	1&c&d\\
	% \end{pmatrix}\begin{pmatrix}
	% 	y\\
	% 	x\\
	% 	z\\
	% \end{pmatrix} = \begin{pmatrix}
	% 	0\\
	% 	ax\\
	% 	0\\
	% \end{pmatrix}.
	% \]

% \[
%  \begin{pmatrix}
% 		1&1&1\\
% 		1&0&b\\
% 		1&c&d\\
% 	\end{pmatrix}\begin{pmatrix}
% 		y\\
% 		x\\
% 		z\\
% 	\end{pmatrix}= 
%     \begin{pmatrix}
% 		0\\
% 		0\\
% 		0\\
% 	\end{pmatrix}
% \]
%         Applying Lemma \ref{lemma 2variable} to the above system gives us the following condition:	
	\[ \begin{vmatrix}
		1&1&1\\
		1&0&b\\
		1&c&d\\
	\end{vmatrix}=0 \implies bc=b+c+d.
	\]
	
 \noindent \textbf{Subcase 2.3:} Suppose that $M_1t=[0, 0, cx]^T$. Applying the same argument as in Subcase~1.1 yields 	
% 	\[
% 	M_1t= \begin{pmatrix}
% 		1&1&1\\
% 		1&a&b\\
% 		1&c&d\\
% 	\end{pmatrix}\begin{pmatrix}
% 		y\\
% 		x\\
% 		z\\
% 	\end{pmatrix} =\begin{pmatrix}
% 		0\\
% 		0\\
% 		cx\\
% 	\end{pmatrix}.
% 	\]
% \[
% \implies \begin{pmatrix}
% 		1&1&1\\
% 		1&a&b\\
% 		1&0&d\\
% 	\end{pmatrix}\begin{pmatrix}
% 		y\\
% 		x\\
% 		z\\
% 	\end{pmatrix}= 
%     \begin{pmatrix}
% 		0\\
% 		0\\
% 		0\\
% 	\end{pmatrix}
% \]
% 	 Applying Lemma \ref{lemma 2variable} to the above system gives us the following condition:	
	\[  \begin{vmatrix}
		1&1&1\\
		1&a&b\\
		1&0&d\\
	\end{vmatrix}=0 \implies ad=a+b+d.
	\]
	
\noindent \textbf{Case 3:} $u = [0, 0, x]^T$, $v = [y, z, 0]^T$, and $t = [y, z, x]^T$.  Then we have,
\[
\begin{aligned}
&M_1u=[x,bx,dx]^T, M_1v=[y+z, y+az, y+cz]^T,\\
&M_1t=[x+y+z, bx+y+az, dx+y+cz]^T.
\end{aligned}
\]

Again, we have $3$ possible subcases.  \\
		
\noindent \textbf{Subcase 3.1:} Suppose that $M_1t=[x, 0, 0]^T$, this will give us the following system of linear equations:	
	
	\[
	M_1t= \begin{pmatrix}
		1&1&1\\
		1&a&b\\
		1&c&d\\
	\end{pmatrix}\begin{pmatrix}
		y\\
		z\\
		x\\
	\end{pmatrix} = \begin{pmatrix}
		x\\
		0\\
		0\\
	\end{pmatrix}
	\]

    \[
\implies \begin{pmatrix}
		1&1&0\\
		1&a&b\\
		1&c&d\\
	\end{pmatrix}\begin{pmatrix}
		y\\
		z\\
		x\\
	\end{pmatrix}= 
    \begin{pmatrix}
		0\\
		0\\
		0\\
	\end{pmatrix}.
\]
        Again, applying Lemma \ref{lemma 2variable} to the above system gives us the following condition:	
	\[  \begin{vmatrix}
		1&1&0\\
		1&a&b\\
		1&c&d\\
	\end{vmatrix}=0 \implies b+d=ad+bc.
	\]
	
 \noindent \textbf{Subcase 3.2:} Suppose that $M_1t=[0, bx, 0]^T$. Proceeding as in Subcase~1.1 and applying Lemma~\ref{lemma 2variable}, we obtain
	
% 	\[
% 	M_1t= \begin{pmatrix}
% 		1&1&1\\
% 		1&a&b\\
% 		1&c&d\\
% 	\end{pmatrix}\begin{pmatrix}
% 		y\\
% 		z\\
% 		x\\
% 	\end{pmatrix} = \begin{pmatrix}
% 		0\\
% 		bx\\
% 		0\\
% 	\end{pmatrix}.
% 	\]

%     \[
% \implies \begin{pmatrix}
% 		1&1&1\\
% 		1&a&0\\
% 		1&c&d\\
% 	\end{pmatrix}\begin{pmatrix}
% 		y\\
% 		z\\
% 		x\\
% 	\end{pmatrix}= 
%     \begin{pmatrix}
% 		0\\
% 		0\\
% 		0\\
% 	\end{pmatrix}
% \]
       
%          Again, applying Lemma \ref{lemma 2variable} to the above system gives us the following condition:	
	
\[  \begin{vmatrix}
		1&1&1\\
		1&a&0\\
		1&c&d\\
	\end{vmatrix}=0 \implies ad=a+c+d.
	\]
	
 \noindent	 \textbf{Subcase 3.3:} Suppose that $M_1t=[0, 0, dx]^T$. Applying the same argument as in Subcase~1.1 yields 
	
% 	\[
% 	M_1t= \begin{pmatrix}
% 		1&1&1\\
% 		1&a&b\\
% 		1&c&d\\
% 	\end{pmatrix}\begin{pmatrix}
% 		y\\
% 		z\\
% 		x\\
% 	\end{pmatrix} = \begin{pmatrix}
% 		0\\
% 		0\\
% 		dx\\
% 	\end{pmatrix}.
% 	\]

%     \[
% \implies \begin{pmatrix}
% 		1&1&1\\
% 		1&a&b\\
% 		1&c&0\\
% 	\end{pmatrix}\begin{pmatrix}
% 		y\\
% 		z\\
% 		x\\
% 	\end{pmatrix}= 
%     \begin{pmatrix}
% 		0\\
% 		0\\
% 		0\\
% 	\end{pmatrix}
% \]
% 	  Again, applying Lemma \ref{lemma 2variable} to the above system gives us the following condition:	
	\[  \begin{vmatrix}
		1&1&1\\
		1&a&b\\
		1&c&0\\
	\end{vmatrix}=0 \implies bc=a+b+c.
	\]
	This triplet $\bigl(\operatorname{wt}(u),\operatorname{wt}(v),\operatorname{wt}(t)\bigr)=(1,2,3)$ gives us nine conditions given below. 
    \[
    a+b+c+d=0,
		\quad ad+bc =\left\{ 
		\begin{aligned}
			a+b\\
			a+c\\
			b+d\\
			c+d\\ 
		\end{aligned}
		\right. ,
		\quad ad =\left\{ 
		\begin{aligned}
			a+b+d\\
			a+c+d\\
		\end{aligned}
		\right. , \]
        \[
		\quad bc= \left\{ 
		\begin{aligned}
			a+b+c\\
			b+c+d\\
		\end{aligned}
		\right. .
    \]

\medskip
\noindent\textbf{Main Case II:}
$\bigl(\operatorname{wt}(u),\operatorname{wt}(v),\operatorname{wt}(t)\bigr)=(2,2,2)$.\\
Since $\bigl(\operatorname{wt}(u),\operatorname{wt}(v),\operatorname{wt}(t)\bigr)=(2,2,2)$, then using MDS property of $M_1$, we have $\operatorname{wt}(M_1u)\geq 2$, $\operatorname{wt}(M_1v)\geq 2$ and $\operatorname{wt}(M_1t)\geq 2$. Thus, from Lemma~\ref{lemma 1}, $(u, M_1u), (v,M_1v)$ form related differentials if and only if $\operatorname{wt}(M_1u)=2$, $\operatorname{wt}(M_1v)=2$, and $\operatorname{wt}(M_1t)=2$. We have the following cases for $(u,v,t) \in \mathbb{F}_{2^m}^3$.\\

    \textbf{Case 1.} $u=[x,y,0]^T$, $v=[0,y,z]^T$, and $t=[x,0,z]^T$ 
    
    \textbf{Case 2.} $u=[0,y,z]^T$, $v=[x,0,z]^T$, and $t=[x,y,0]^T$ 
    
    \textbf{Case 3.} $u=[x,0,z]^T$, $v=[x,y,0]^T$, and $t=[0,y,z]^T$.

We observe that the triplets in Cases 2 and 3 are merely permutations of the triplet (u,v,t) given in Case 1. Hence, the conditions derived from these cases are equivalent. Therefore, it suffices to consider only one representative triplet (u,v,t) from any of the three cases. Let us consider Case 1. \\
\noindent \textbf{Case 1:}  $u=[x,y,0]^T, v=[0,y,z]^T$ and $t = [x, 0, z]^T$.  Then we have
\begin{equation*}
     M_1u=[x+y,x+ay,x+cy]^T, M_1v=[y+z, ay+bz,cy+dz]^T, 
\end{equation*}
\begin{equation}\label{eq: 222}
    M_1t=[x+z,x+bz,x+dz]^T.
\end{equation}

    Since $\bigl(\operatorname{wt}(M_1u),\operatorname{wt}(M_1v),\operatorname{wt}(M_1t)\bigr)=(2,2,2)$, the related differences property implies that each of $M_1u$, $M_1v$, and $M_1t$ has a zero in one coordinate, and these zero coordinates must occur at distinct positions. Consequently, there are
\[
\binom{3}{1}\binom{2}{1}\binom{1}{1}=6
\]
possible subcases. We discuss each subcase one by one.  

    \noindent \textbf{Subcase 1.1:} Suppose that \[M_1u=[0,x+ay,x+cy]^T, M_1v=[y+z, 0,cy+dz]^T,\] 
    \[ M_1t=[x+z,x+bz,0]^T.\] By comparing it with (\ref{eq: 222}), we get $x+y=0, ay+bz=0$ and $x+dz=0$.
    This will give us the following system of linear equations:	
	
	\[
	\begin{pmatrix}
		1&1&0\\
		0&a&b\\
		1&0&d\\
	\end{pmatrix}\begin{pmatrix}
		x\\
		y\\
		z\\
	\end{pmatrix} = \begin{pmatrix}
		0\\
		0\\
		0\\
	\end{pmatrix}.
	\]
    
    Upon applying Lemma \ref{lemma 2variable} to the above system gives us the following condition:	
	\[  \begin{vmatrix}
		1&1&0\\
		0&a&b\\
		1&0&d\\
	\end{vmatrix}=0 \implies ad+b=0.
	\]

    \noindent \textbf{Subcase 1.2:} Suppose that \[M_1u=[0,x+ay,x+cy]^T, M_1v=[y+z, ay+bz,0]^T, \] \[ M_1t=[x+z,0,x+dz]^T.\] 
    Proceeding as in Subcase~1.1 and applying Lemma~\ref{lemma 2variable}, we obtain

    %By comparing it with (\ref{eq: 222}), we get $x+y=0, cy+dz=0$ and $x+bz=0$.
    %This will give us the following system of linear equations:	
	
	% \[
	% \begin{pmatrix}
	% 	1&1&0\\
	% 	1&0&b\\
	% 	0&c&d\\
	% \end{pmatrix}\begin{pmatrix}
	% 	x\\
	% 	y\\
	% 	z\\
	% \end{pmatrix} = \begin{pmatrix}
	% 	0\\
	% 	0\\
	% 	0\\
	% \end{pmatrix}.
	% \]
    
 %    Upon applying Lemma \ref{lemma 2variable} to the above system gives us the following condition:	
	\[  \begin{vmatrix}
		1&1&0\\
		1&0&b\\
		0&c&d\\
	\end{vmatrix}=0 \implies d+bc=0.
	\]

    \noindent \textbf{Subcase 1.3:} Suppose that \[M_1u=[x+y,0,x+cy]^T, M_1v=[0, ay+bz,cy+dz]^T, \] \[ M_1t=[x+z,x+bz,0]^T.\]
    Applying the same argument as in Subcase~1.1 yields
    
    %By comparing it with (\ref{eq: 222}), we get $x+ay=0, y+z=0$ and $x+dz=0$. This will give us the following system of linear equations:	
	
	% \[
	% \begin{pmatrix}
	% 	0&1&1\\
	% 	1&a&0\\
	% 	1&0&d\\
	% \end{pmatrix}\begin{pmatrix}
	% 	x\\
	% 	y\\
	% 	z\\
	% \end{pmatrix} = \begin{pmatrix}
	% 	0\\
	% 	0\\
	% 	0\\
	% \end{pmatrix}.
	% \]
    
 %    Applying Lemma \ref{lemma 2variable} to the above system gives us the following condition:	
	\[ \begin{vmatrix}
		0&1&1\\
		1&a&0\\
		1&0&d\\
	\end{vmatrix}=0 \implies a+d=0.
	\]

    \noindent \textbf{Subcase 1.4:} Suppose that \[M_1u=[x+y,0,x+cy]^T, M_1v=[y+z, ay+bz,0]^T,\] 
    \[M_1t=[0,x+bz,x+dz]^T.\] 
    Applying the same argument as in Subcase~1.1 yields
    
    %By comparing it with (\ref{eq: 222}), we get $x+ay=0, cy+dz=0$ and $x+z=0$.
    %This will give us the following system of linear equations:	
	
	% \[
	% \begin{pmatrix}
	% 	1&0&1\\
	% 	1&a&0\\
	% 	0&c&d\\
	% \end{pmatrix}\begin{pmatrix}
	% 	x\\
	% 	y\\
	% 	z\\
	% \end{pmatrix} = \begin{pmatrix}
	% 	0\\
	% 	0\\
	% 	0\\
	% \end{pmatrix}.
	% \]
    
 %    Upon applying Lemma \ref{lemma 2variable} to the above system gives us the following condition:	
	\[ \begin{vmatrix}
		1&0&1\\
		1&a&0\\
		0&c&d\\
	\end{vmatrix}=0 \implies ad+c=0.
	\]

    \noindent \textbf{Subcase 1.5:} Suppose that \[M_1u=[x+y,x+ay,0]^T, M_1v=[0, ay+bz,cy+dz]^T, \] 
    \[M_1t=[x+z,0,x+dz]^T.\] 
    Applying the same argument as in Subcase~1.1 yields
    
    %By comparing it with (\ref{eq: 222}), we get $x+cy=0, y+z=0$ and $x+bz=0$.
    %This will give us the following system of linear equations:	
	
	% \[
	% \begin{pmatrix}
	% 	0&1&1\\
	% 	1&0&b\\
	% 	1&c&0\\
	% \end{pmatrix}\begin{pmatrix}
	% 	x\\
	% 	y\\
	% 	z\\
	% \end{pmatrix} = \begin{pmatrix}
	% 	0\\
	% 	0\\
	% 	0\\
	% \end{pmatrix}.
	% \]
    
 %    Upon applying Lemma \ref{lemma 2variable} to the above system gives us the following condition:	
	\[ \begin{vmatrix}
		0&1&1\\
		1&0&b\\
		1&c&0\\
	\end{vmatrix}=0 \implies b+c=0.
	\]

    \noindent \textbf{Subcase 1.6:} Suppose that \[M_1u=[x+y,x+ay,0]^T, M_1v=[y+z, 0,cy+dz]^T, \] 
    \[M_1t=[0,x+bz, x+dz]^T.\] 
    Applying the same argument as in Subcase~1.1 yields
    
    %By comparing it with (\ref{eq: 222}), we get $x+cy=0, ay+bz=0$ and $x+z=0$.
    %This will give us the following system of linear equations:	
	
	% \[
	% \begin{pmatrix}
	% 	1&0&1\\
	% 	0&a&b\\
	% 	1&c&0\\
	% \end{pmatrix}\begin{pmatrix}
	% 	x\\
	% 	y\\
	% 	z\\
	% \end{pmatrix} = \begin{pmatrix}
	% 	0\\
	% 	0\\
	% 	0\\
	% \end{pmatrix}.
	% \]
    
 %    Upon applying Lemma \ref{lemma 2variable} to the above system gives us the following condition:	
	\[ \begin{vmatrix}
		1&0&1\\
		0&a&b\\
		1&c&0\\
	\end{vmatrix}=0 \implies a+bc=0.
	\]
    Thus, this triplet $\bigl(\operatorname{wt}(u),\operatorname{wt}(v),\operatorname{wt}(t)\bigr)=(2,2,2)$ gives us six conditions given below.
    \[
    \quad ad= \left\{ 
		\begin{aligned}
			b\\
			c\\
		\end{aligned}
		\right. ,
    \quad bc= \left\{ 
		\begin{aligned}
			a\\
			d\\
		\end{aligned}
		\right. ,
    \quad 0= \left\{ 
		\begin{aligned}
			a+d\\
			b+c\\
		\end{aligned}
		\right. .
    \]
	This completes the proof.
\end{proof}

\begin{remark}\label{Remark:PG(3,2)}
Each condition in Theorem~\ref{thm3*3} has the form $C_i=Q_i+L_i=0$, where $Q_i\in\{0,\,ad,\,bc,\,ad+bc\}$ and $L_i$ is a nonzero $\mathbb{F}_2$-linear form in $a,b,c,d$. Mapping $C_i$ to the coefficient vector of $L_i$ gives a bijection $\varphi:\{C_1,\ldots,C_{15}\} \longrightarrow \mathbb{F}_2^4\setminus\{0\}$, which is naturally identified with the $15$ points of the projective space $\mathrm{PG}(3,2)$ of dimension $3$ over $\mathbb{F}_2$. More explicitly, if $L_i=\varepsilon_1a+\varepsilon_2b+\varepsilon_3c+\varepsilon_4d$, then the corresponding quadratic part is $Q_i=(\varepsilon_2+\varepsilon_3)ad +(\varepsilon_1+\varepsilon_4)bc$, so the quadratic part is entirely determined by the linear part. Under this identification, the $15$ points have Hamming-weight distribution $4,6,4,1$ over weights $1,2,3,4$. Relative to the frame $(a,b,c,d)$, the nine conditions arising from the $(1,2,3)$ case comprise $1,4,4$ points of weights $4,3,2$, respectively, whereas the six conditions arising from the $(2,2,2)$ case comprise $4,2$ points of weights $1,2$, respectively. This coordinate-dependent description gives a finite-geometric organization of the exhaustive conditions derived
from Lemmas~\ref{lemma_poss_triplet} and~\ref{lemma_exclude_triplet} and the case analysis in Theorem~\ref{thm3*3}.
\end{remark}

\begin{remark}
The characterization of $3\times 3$ MDS matrices without related differentials was also studied by Otal et al.~\cite[Theorem~6]{otal2025cryptographic}. Following our communication concerning the conditions in the theorem statement, the authors revised the statement and acknowledged our observation in Remark~3 of their updated manuscript. There, they clarify that the proof was complete and that two conditions had been inadvertently omitted from the statement. The characterization in Theorem~\ref{thm3*3} agrees with their revised statement, but our derivation follows a different approach. We further apply this characterization to enumerate $3\times 3$ MDS matrices without related differentials.
\end{remark}

\noindent The complete characterization in Theorem~\ref{thm3*3} yields an efficient algorithm, shown in Algorithm~\ref{alg:rd-verification}, for verifying whether a given $3 \times 3$ MDS matrix $M$ has related differentials. A SageMath implementation of this algorithm is provided in Appendix~\ref{app:sage}.

\begin{algorithm}[htb]
	\caption{Verification of Related Differentials for $3 \times 3$ MDS Matrices}
	\label{alg:rd-verification}
	\begin{algorithmic}[1]
		\Require Finite field $\mathbb{F}_{2^m}$ and a $3\times 3$ MDS matrix $M=(m_{ij})$
		\Ensure Returns whether the matrix avoids or has related differentials
		\State $D_1 \gets \operatorname{Diag}(m_{11}, m_{21}, m_{31})$ and $D_2 \gets \operatorname{Diag}\!\left(1, m_{11}^{-1}m_{12}, m_{11}^{-1}m_{13}\right)$
		\State $M_1 \gets D_1^{-1} M D_2^{-1} =
		\begin{pmatrix}
			1 & 1 & 1 \\
			1 & a & b \\
			1 & c & d
		\end{pmatrix}$
		\For{$i \gets 1$ to $15$}
		\If{Condition $C_i(a,b,c,d)$ in \eqref{Eqn:3MDS_RD} is satisfied}
		\State \Return Has Related Differentials
		\EndIf
		\EndFor
		\State \Return Avoids Related Differentials
	\end{algorithmic}
\end{algorithm}

\begin{example}\label{Example:3MDS_noRD}
	% We construct a $3\times 3$ MDS matrix over the finite field $\mathbb{F}_{2^4}$ that avoids related differentials by ensuring none of the fifteen conditions are satisfied. 
	Consider the matrix
	\[
	M_1=
	\begin{pmatrix}
		1  &  1  &  1\\
		1  &  \alpha    &   \alpha^2\\
		1  &  \alpha^3  &   \alpha^3 + \alpha^2
	\end{pmatrix},
	\]
	where $\alpha$ is a primitive element and root of the constructing polynomial $x^4+x+1$. It can be verified that $M_1$ is an MDS matrix, also it does not satisfy the required equations given in \eqref{Eqn:3MDS_RD}. Therefore, $M_1$ avoids related differentials.
\end{example}

\begin{remark}
	In Theorem~\ref{thm3*3}, we characterize the conditions under which a $3\times 3$ MDS matrix avoids related differentials, and in Example~\ref{Example:3MDS_noRD}, we provide a concrete example of such a matrix. Involutory MDS matrices form another important family studied in the literature, so it is natural to ask whether a $3\times 3$ involutory MDS matrix can also avoid related differentials. However, in~\cite{otal2025cryptographic}, the authors show that $3\times 3$ involutory MDS matrices have related differentials because of their structural properties.
\end{remark}

\noindent We conclude this section with enumeration results on $3\times 3$ MDS matrices with no related differentials. In~\cite{Kumar_MDS2024}, the authors derive the closed form $(2^m-1)^5(2^m-2)(2^m-3)(2^{2m}-9\cdot 2^m+21)$ for the number of $3\times 3$ MDS matrices over $\mathbb{F}_{2^m}$. Based on the characterization in Theorem~\ref{thm3*3}, we compute the number of $3\times 3$ MDS matrices that have no pair of related differentials. Table~\ref{tab:mds_enumeration} compares this refined count with the total number of $3\times 3$ MDS matrices, showing that excluding matrices that admit related differentials reduces the set of admissible matrices. 

The counts in Table~\ref{tab:mds_enumeration} enumerate all $3\times 3$ MDS matrices that have no pair of related differentials. The enumeration is performed by evaluating the $15$ polynomial conditions of~\eqref{Eqn:3MDS_RD} over all parameter tuples $(a,b,c,d)$ that yield a valid representative MDS matrix~$M_1$, and then scaling by the $(2^m-1)^5$ factor arising from the diagonal equivalence of Theorem~\ref{Th_decomposition} and the invariance under nonsingular diagonal multiplication established in Lemma~\ref{Diag_RD}. The proportion of RD-free matrices increases monotonically with the field size. Over $\mathbb{F}_{2^3}$, no $3\times 3$ MDS matrix avoids related differentials, yielding an RD-free proportion of $0\%$. This proportion increases to approximately $18.44\%$ over $\mathbb{F}_{2^4}$, $54.88\%$ over $\mathbb{F}_{2^5}$, and $94.16\%$ over $\mathbb{F}_{2^8}$. Thus, as field size increases, MDS matrices resistant to related differentials become increasingly prevalent. 
% This trend confirms the practical relevance of the characterization established in Theorem~\ref{thm3*3}.

% \begin{table}[ht]
% 	\centering
% 	\caption{Enumeration of $3\times 3$ MDS matrices and their subclass that admit no related differential (RD) pair over finite fields.}
% 	\label{tab:mds_enumeration}
% 	% Set row stretching for a cleaner look
% 	% \renewcommand{\arraystretch}{1.2} 
% 	\resizebox{0.8\textwidth}{!}{%
% 		\begin{tabular}{c|c|c}
% 			\toprule
% 			{Finite Field} & {\# MDS} &  \makecell{{\# MDS with} {no RD pair} }
% 			\\
% 			\midrule
% 			% $\mathbb{F}_{2^3}$ & $7^5 \times 390$ & $7^2\times 24$ & 0 & 0 \\
% 			% $\mathbb{F}_{2^4}$ & $15^5\times 24206$ & $15^2\times 168$ & $15^5\times 4464$ & \\
% 			% $\mathbb{F}_{2^5}$ & $31^5\times 658590$ & $31^2\times 840$ & $31^5\times 361440$  &  \\
% 			% $\mathbb{F}_{2^6}$ & $63^5\times 13392062$ & $63^2\times 3720$ & $63^5 \times 10298160$  &  \\
% 			% $\mathbb{F}_{2^7}$ & $127^5\times 240234750$ & $127^2\times 15624$ & $127^5\times 212254560$ &  \\
% 			% $\mathbb{F}_{2^8}$ & $255^5\times 4064764286$ & $255^2\times 64008$ & $255^5\times 3827268144$ &  \\
% 			$\mathbb{F}_{2^3}$ & $7^5 \times 390$ & 0 \\
% 			$\mathbb{F}_{2^4}$ & $15^5\times 24206$ & $15^5\times 4464$ \\
% 			$\mathbb{F}_{2^5}$ & $31^5\times 658590$ & $31^5\times 361440$ \\
% 			$\mathbb{F}_{2^6}$ & $63^5\times 13392062$ & $63^5\times 10298160$ \\
% 			$\mathbb{F}_{2^7}$ & $127^5\times 240234750$ & $127^5\times 212254560$ \\
% 			$\mathbb{F}_{2^8}$ & $255^5\times 4064764286$ & $255^5\times 3827268144$ \\
% 			\bottomrule
% 		\end{tabular}
% 	}
% \end{table}

\begin{table*}[ht]
	\centering
	\caption{Enumeration of $3\times 3$ MDS matrices and their subclass that admit no related differential (RD) pair over finite fields.}
	\label{tab:mds_enumeration}
	% Set row stretching for a cleaner look
	% \renewcommand{\arraystretch}{1.2} 
	\resizebox{0.8\textwidth}{!}{%
		\begin{tabular}{c|c|c|c}
			\toprule
			{Finite Field} & {\# MDS} &  \makecell{{\# MDS with} {no RD pair} } & {\% RD-free}
			\\
			\midrule
			$\mathbb{F}_{2^3}$ & $7^5 \times 390$ & 0 & $0.00\%$ \\
			$\mathbb{F}_{2^4}$ & $15^5\times 24206$ & $15^5\times 4464$ & $18.44\%$ \\
			$\mathbb{F}_{2^5}$ & $31^5\times 658590$ & $31^5\times 361440$ & $54.88\%$ \\
			$\mathbb{F}_{2^6}$ & $63^5\times 13392062$ & $63^5\times 10298160$ & $76.90\%$ \\
			$\mathbb{F}_{2^7}$ & $127^5\times 240234750$ & $127^5\times 212254560$ & $88.35\%$ \\
			$\mathbb{F}_{2^8}$ & $255^5\times 4064764286$ & $255^5\times 3827268144$ & $94.16\%$ \\
			\bottomrule
		\end{tabular}
	}
\end{table*}

\section{Conclusion}\label{Sec:Conclusion}
In this paper, we study the existence of related differentials in linear layers over finite fields and analyze how this property depends on the underlying matrix structure. We prove that the MDS property is necessary for avoiding related differentials, by showing that every nonsingular non-MDS matrix admits related differentials. We then investigate two important structured classes of MDS matrices. For symmetric matrices, we prove that all odd order symmetric MDS matrices possess related differentials. For circulant matrices, we show that circulant MDS matrices of order $n$ admit related differentials whenever $n \not\equiv \pm 2 \pmod{12}$, yielding a broad characterization of related-differential behavior in this family. We also revisit the $3\times 3$ case and derive a complete characterization through necessary and sufficient conditions, enabling exhaustive verification and the construction of MDS matrices resistant to related-differential cryptanalysis.

These results also point to several directions for future research. Existing studies have addressed related differentials for structured $4\times4$ MDS matrices, including circulant matrices~\cite{daemen2009new} and Hadamard matrices~\cite{jha2025construction}. A natural next step is therefore to obtain a complete characterization for general $4\times4$ MDS matrices beyond these structured classes. Since the present work is primarily theoretical, another important direction is the construction of lightweight MDS matrices without related differentials and with practical efficiency constraints. It would also be interesting to study other structured families beyond circulant and Hadamard matrices and determine whether they exhibit similar limitations or admit new resistant constructions. Finally, in the circulant setting, our results suggest a natural conjecture for the remaining unresolved cases, and proving or disproving this conjecture remains an interesting open problem.

\appendices

\section{SageMath Verification Code}\label{app:sage}
The following SageMath function implements Algorithm~\ref{alg:rd-verification} by evaluating the $15$ conditions of~\eqref{Eqn:3MDS_RD} directly. It returns \texttt{True} if the input matrix admits related differentials and \texttt{False} otherwise.

% \begin{verbatim}
% def has_rd(M):
%     D1 = diagonal_matrix([M[0,0], M[1,0], M[2,0]])
%     D2 = diagonal_matrix(
%         [1, M[0,1]/M[0,0], M[0,2]/M[0,0]])
%     R = D1.inverse() * M * D2.inverse()
%     a, b, c, d = R[1,1], R[1,2], R[2,1], R[2,2]

%     C = [a+b+c+d, a+d, b+c,
%          a*d+b*c+a+b, a*d+b*c+a+c,
%          a*d+b*c+b+d, a*d+b*c+c+d,
%          a*d+a+b+d, a*d+a+c+d,
%          a*d+b, a*d+c,
%          b*c+a+b+c, b*c+b+c+d,
%          b*c+a, b*c+d]

%     return any(f == 0 for f in C)
% \end{verbatim}

\begin{lstlisting}[style=ieeepy]
def has_rd(M):
    D1 = diagonal_matrix([M[0,0], M[1,0], M[2,0]])
    D2 = diagonal_matrix(
        [1, M[0,1]/M[0,0], M[0,2]/M[0,0]])
    R = D1.inverse() * M * D2.inverse()
    a, b, c, d = R[1,1], R[1,2], R[2,1], R[2,2]

    C = [a+b+c+d, a+d, b+c,
         a*d+b*c+a+b, a*d+b*c+a+c,
         a*d+b*c+b+d, a*d+b*c+c+d,
         a*d+a+b+d, a*d+a+c+d,
         a*d+b, a*d+c,
         b*c+a+b+c, b*c+b+c+d,
         b*c+a, b*c+d]

    return any(f == 0 for f in C)
\end{lstlisting}

Using the above function, we enumerate all $3\times 3$ MDS matrices over small fields and compute the proportion that avoid related differentials. Table~\ref{tab:mds_enumeration} reports the total number of $3\times 3$ MDS matrices and the number that admit no related-differential pair. These counts are obtained by multiplying the number of representative matrices $M_1$ by $(2^m-1)^5$, accounting for diagonal equivalence as established in Theorem~\ref{Th_decomposition} and Lemma~\ref{Diag_RD}.

\section*{Acknowledgment} 
We are thankful to the anonymous reviewers for their valuable comments, which have substantially improved the manuscript. We also thank Kamil Otal for helpful discussions concerning \cite{otal2025cryptographic} and the characterization of $3\times 3$ MDS matrices without related differentials.

% \bstctlcite{IEEEtranBSTControl}
\bibliographystyle{IEEEtran}  % Use IEEE format
\bibliography{ref.bib}

@IEEEtranBSTCTL{IEEEtranBSTControl,
  CTLdash_repeated_names = "no"
}

@inproceedings{kc2,
    title={On constructions of involutory {MDS} matrices},
    author={Gupta, Kishan Chand and Ray, Indranil Ghosh},
    booktitle={International Conference on Cryptology in Africa},
    pages={43--60},
    year={2013},
    url={https://doi.org/10.1007/978-3-642-38553-7_3},
    DOI={10.1007/978-3-642-38553-7_3},
    organization={Springer}
}

@article{kcz,
    title = {Cryptographically significant {MDS} matrices over finite fields: {A} brief survey and some generalized results},
    journal = {Advances in Mathematics of Communications},
    volume = {13},
    number = {4},
    pages = {779-843},
    year = {2019},
    url={https://doi.org/10.3934/amc.2019045},
    DOI={10.3934/amc.2019045},
    author = {Kishan Chand Gupta and
          Sumit Kumar Pandey and
          Indranil Ghosh Ray and
          Susanta Samanta}
}

@article{sdm,
    author = {Sajadieh, Mahdi and Dakhilalian, Mohammad and Mala, Hamid and Omoomi, Behnaz},
    title = {On Construction of {I}nvolutory {MDS} {M}atrices from {V}andermonde {M}atrices in {$GF(2^q)$}},
    year = {2012},
    issue_date = {September 2012},
    publisher = {Kluwer Academic Publishers},
    address = {USA},
    volume = {64},
    number = {3},
    issn = {0925-1022},
    journal = {Designs, Codes and Cryptography},
    month = {sep},
    url={https://doi.org/10.1007/s10623-011-9578-x},
    DOI={10.1007/s10623-011-9578-x},
    pages = {287-308}
}

@article{psa,
    author = {Pehlivano{\~{g}}lu, Meltem Kurt and Sakalli, Muharrem Tolga and Akleylek, Sedat and Duru, Nevcihan and Rijmen, Vincent},
    title = {Generalisation of {H}adamard matrix to generate involutory {MDS} matrices for lightweight cryptography},
    journal = {IET Information Security},
    volume = {12},
    url={https://doi.org/10.1049/iet-ifs.2017.0156},
    DOI= {10.1049/iet-ifs.2017.0156},
    number = {4},
    pages = {348-355},
    year = {2018}
}

@article{gmt,
    title={A new matrix form to generate all $3 \times 3$ involutory {MDS} matrices over $\mathbb{F}_{2^m}$},
    author={G{\"u}zel, G{\"u}ls{\"u}m G{\"o}zde and Sakalli, Muharrem Tolga and Akleylek, Sedat and Rijmen, Vincent and {\c{C}}engellenmi{\c{s}}, Yasemin},
    journal={Information {P}rocessing {L}etters},
    volume={147},
    pages={61--68},
    year={2019},
    url={https://doi.org/10.1016/j.ipl.2019.02.013},
    DOI={10.1016/j.ipl.2019.02.013},
    publisher={Elsevier}
}

@book{JDA_Thesis_1995,
    author       = {Joan Daemen},
    title        = {Cipher and hash function design, strategies based on linear and differential cryptanalysis, {PhD} {T}hesis},
    publisher    = {K.U.Leuven},
    url= {https://cs.ru.nl/~joan/papers/JDA_Thesis_1995.pdf},
    year         = {1995}
}

@InProceedings{LACAN2003,
    author="Lacan, J{\'e}r{\^o}me
    and Fimes, J{\'e}r{\^o}me",
    editor="Mullen, Gary L.
    and Poli, Alain
    and Stichtenoth, Henning",
    title="A {C}onstruction of {M}atrices with {N}o {S}ingular {S}quare {S}ubmatrices",
    booktitle="Finite Fields and Applications",
    year="2004",
    publisher="Springer Berlin Heidelberg",
    address="Berlin, Heidelberg",
    pages="145--147",
    url={https://doi.org/10.1007/978-3-540-24633-6_11},
    DOI= {10.1007/978-3-540-24633-6_11},
    isbn="978-3-540-24633-6"
}

@article{Gupta2023direct,
    title = {{On the Direct Construction of MDS and Near-MDS Matrices}},
    journal = {Advances in Mathematics of Communications},
    volume = {24},
    number = {0},
    pages = {110--135},
    year = {2026},
    issn = {1930-5346},
    doi = {10.3934/amc.2026030},
    url = {https://doi.org/10.3934/amc.2026030},
    author = {Kishan Chand Gupta and Sumit Kumar Pandey and Susanta Samanta}
}

@book{AES,
    author    = {Joan Daemen and
               Vincent Rijmen},
    title     = {The {D}esign of {R}ijndael: {AES} - {T}he {A}dvanced {E}ncryption {S}tandard},
    series    = {Information Security and Cryptography},
    publisher = {Springer},
    year      = {2002},
    url       = {https://doi.org/10.1007/978-3-662-04722-4},
    doi       = {10.1007/978-3-662-04722-4},
    isbn      = {3-540-42580-2},
    bibsource = {dblp computer science bibliography, https://dblp.org}
}

@book{FJ77,
  author = {MacWilliams, F.J. and Sloane, N.J.A.},
  title = {The {T}heory of {E}rror {C}orrecting {C}odes},
  year = {1977},
  publisher = {North-Holland Publishing Co., Amsterdam-New York-Oxford},
  url= {https://shop.elsevier.com/books/the-theory-of-error-correcting-codes/macwilliams/978-0-444-85193-2}
}

@article{Kumar_MDS2024,
    title = {Construction of all {MDS} and involutory {MDS} matrices},
    journal = {Advances in Mathematics of Communications},
    volume = {19},
    number = {3},
    pages = {922--941},
    year = {2025},
    author = {Yogesh Kumar and Prasanna Raghaw Mishra and Susanta Samanta and Kishan Chand Gupta and Atul Gaur},
    url={https://doi.org/10.3934/amc.2024033},
    DOI={10.3934/amc.2024033}
}

@article{Kumar_4MDS,
  title = {A systematic construction approach for all $4\times 4$ involutory {MDS} matrices},
  author = {Yogesh Kumar and P. R. Mishra and Susanta Samanta and Atul Gaur},
  journal = {Journal of Applied Mathematics and Computing},
  year = {2024},
  volume = {70},
  number = {5},
  pages = {4677--4697},
  isbn = {1865-2085},
  doi = {10.1007/s12190-024-02142-z},
  url = {https://doi.org/10.1007/s12190-024-02142-z}
}

@article{jha2025construction,
  title={Construction of Hadamard-based MixColumns Matrices Resistant to Related-Differential Cryptanalysis},
  author={Jha, Sonu and Li, Shun and Gligoroski, Danilo},
  journal={IACR Communications in Cryptology},
  volume={2},
  number={1},
  year={2025},
  url= {https://doi.org/10.62056/a6ksdk5vt},
  DOI= {10.62056/a6ksdk5vt}
}

@article{daemen2009new,
  title={New criteria for linear maps in {AES}-like ciphers},
  author={Daemen, Joan and Rijmen, Vincent},
  journal={Cryptography and Communications},
  volume={1},
  number={1},
  pages={47--69},
  year={2009},
  url={https://doi.org/10.1007/s12095-008-0003-x},
  DOI={10.1007/s12095-008-0003-x},
  publisher={Springer}
}

@article{GhaedBardeh_Rijmen_2022,
  title={New Key-Recovery Attack on Reduced-Round AES}, 
  author={Ghaedi Bardeh, Navid and Rijmen, Vincent}, 
  journal={IACR Transactions on Symmetric Cryptology}, 
  url={https://tosc.iacr.org/index.php/ToSC/article/view/9713}, 
  DOI={10.46586/tosc.v2022.i2.43-62},
  volume={2022},
  number={2}, 
  year={2022}, 
  month={Jun.}, 
  pages={43–62} 
}

@inproceedings{Augot2014ShortenedBCH,
  author    = {Daniel Augot and Matthieu Finiasz},
  title     = {{Direct Construction of Recursive MDS Diffusion Layers Using Shortened BCH Codes}},
  booktitle = {Fast Software Encryption -- FSE 2014},
  series    = {Lecture Notes in Computer Science},
  volume    = {8540},
  pages     = {3--17},
  publisher = {Springer},
  year      = {2014},
  address   = {London, UK},
  month     = {March},
  url={https://doi.org/10.1007/978-3-662-46706-0_1},
  doi       = {10.1007/978-3-662-46706-0_1}
}

@article{GuptaPV17_1,
  author    = {Kishan Chand Gupta and
               Sumit Kumar Pandey and
               Ayineedi Venkateswarlu},
  title     = {Towards a general construction of recursive {MDS} diffusion layers},
  journal   = {Designs, Codes and Cryptography},
  volume    = {82},
  number    = {1-2},
  pages     = {179--195},
  year      = {2017},
  url={https://doi.org/10.1007/s10623-016-0261-0},
  doi       = {10.1007/s10623-016-0261-0}
}

@article{GuptaPV17_2,
  title     = {On the direct construction of recursive {MDS} matrices},
  author    = {Kishan Chand Gupta and
               Sumit Kumar Pandey and
               Ayineedi Venkateswarlu},
  journal   = {Designs, Codes and Cryptography},
  volume    = {82},
  number    = {1-2},
  pages     = {77--94},
  year      = {2017},
  url={https://doi.org/10.1007/s10623-016-0233-4},
  doi       = {10.1007/s10623-016-0233-4}
}

@article{Kesarwani_FSE2019, 
  title={Exhaustive {S}earch for {V}arious {T}ypes of {MDS} {M}atrices}, 
  volume={2019}, 
  url={https://tosc.iacr.org/index.php/ToSC/article/view/8364}, 
  DOI={10.13154/tosc.v2019.i3.231-256}, 
  number={3}, 
  journal={IACR Transactions on Symmetric Cryptology}, 
  author={Kesarwani, Abhishek and Sarkar, Santanu and Venkateswarlu, Ayineedi}, 
  year={2019}, 
  month={Sep.}, 
  pages={231-256} 
}

@inproceedings{Sim2015LightweightInvolution,
  author    = {Siang Meng Sim and Khoongming Khoo and Fr{\'e}d{\'e}rique Oggier and Thomas Peyrin},
  title     = {{Lightweight MDS Involution Matrices}},
  booktitle = {Fast Software Encryption -- FSE 2015},
  editor    = {Gregor Leander},
  series    = {Lecture Notes in Computer Science},
  volume    = {9054},
  pages     = {471--493},
  publisher = {Springer Berlin Heidelberg},
  address   = {Berlin, Heidelberg},
  year      = {2015},
  url={https://doi.org/10.1007/978-3-662-48116-5_23},
  doi       = {10.1007/978-3-662-48116-5_23}
}

@inproceedings{Liu2016GeneralizedCirculant,
  author    = {Meicheng Liu and Siang Meng Sim},
  title     = {{Lightweight MDS Generalized Circulant Matrices}},
  booktitle = {Fast Software Encryption -- FSE 2016},
  editor    = {Thomas Peyrin},
  series    = {Lecture Notes in Computer Science},
  volume    = {9783},
  pages     = {101--120},
  publisher = {Springer},
  year      = {2016},
  address   = {Bochum, Germany},
  month     = {March},
  url={https://doi.org/10.1007/978-3-662-52993-5_6},
  doi       = {10.1007/978-3-662-52993-5_6}
}

@inproceedings{Li2016CirculantInvolutory,
  author    = {Yongqiang Li and Mingsheng Wang},
  title     = {On the Construction of Lightweight Circulant Involutory MDS Matrices},
  booktitle = {Fast Software Encryption -- FSE 2016},
  editor    = {Thomas Peyrin},
  series    = {Lecture Notes in Computer Science},
  volume    = {9783},
  pages     = {121--139},
  publisher = {Springer},
  year      = {2016},
  address   = {Bochum, Germany},
  month     = {March},
  url={https://doi.org/10.1007/978-3-662-52993-5_7},
  doi       = {10.1007/978-3-662-52993-5_7}
}

@article{otal2025cryptographic,
  title={{On the Cryptographic Resilience of MDS Matrices}},
  author={Otal, Kamil and S{\"u}l{\c{c}}e, Ali Mert and Yayla, O{\u{g}}uz},
  journal={Cryptology ePrint Archive},
  url={https://ia.cr/2025/2240},
  year={2025}
}

@InProceedings{SHARK,
  author="Rijmen, Vincent
  and Daemen, Joan
  and Preneel, Bart
  and Bosselaers, Antoon
  and De Win, Erik",
  editor="Gollmann, Dieter",
  title="The cipher {SHARK}",
  booktitle="Fast Software Encryption",
  year="1996",
  DOI= {10.1007/3-540-60865-6_47},
  url={https://doi.org/10.1007/3-540-60865-6_47},
  publisher="Springer Berlin Heidelberg",
  address="Berlin, Heidelberg",
  pages="99--111",
  isbn="978-3-540-49652-6"
}

@InProceedings{SQUARE,
  author=         "Daemen, Joan
                  and Knudsen, Lars
                  and Rijmen, Vincent",
  editor=         "Biham, Eli",
  title=          "The block cipher {S}quare",
  booktitle=      "Fast Software Encryption",
  year=           "1997",
  publisher=      "Springer Berlin Heidelberg",
  address=        "Berlin, Heidelberg",
  pages=          "149--165",
  url= {https://doi.org/10.1007/BFb0052343},
  DOI={10.1007/BFb0052343}
}

@inproceedings{RBH17,
  author       = {Sondre R{\o}njom and
                  Navid Ghaedi Bardeh and
                  Tor Helleseth},
  editor       = {Tsuyoshi Takagi and
                  Thomas Peyrin},
  title        = {Yoyo {T}ricks with {AES}},
  booktitle    = {Advances in Cryptology - {ASIACRYPT} 2017 - 23rd International Conference on the Theory and Applications of Cryptology and Information Security, Hong Kong, China, December 3-7, 2017, Proceedings, Part {I}},
  series       = {Lecture Notes in Computer Science},
  pages        = {217--243},
  publisher    = {Springer},
  year         = {2017},
  url          = {https://doi.org/10.1007/978-3-319-70694-8\_8},
  doi          = {10.1007/978-3-319-70694-8\_8}
}

@InProceedings{PRIDE,
author="Albrecht, Martin R.
and Driessen, Benedikt
and Kavun, Elif Bilge
and Leander, Gregor
and Paar, Christof
and Yal{\c{c}}{\i}n, Tolga",
editor="Garay, Juan A.
and Gennaro, Rosario",
title="Block {C}iphers -- {F}ocus on the {L}inear {L}ayer (feat. {PRIDE})",
booktitle="Advances in Cryptology -- CRYPTO 2014",
year="2014",
publisher="Springer Berlin Heidelberg",
address="Berlin, Heidelberg",
pages="57--76",
isbn="978-3-662-44371-2",
url={https://doi.org/10.1007/978-3-662-44371-2_4},
DOI={10.1007/978-3-662-44371-2_4}
}

@InProceedings{Fides,
author="Bilgin, Beg{\"u}l
and Bogdanov, Andrey
and Kne{\v{z}}evi{\'{c}}, Miroslav
and Mendel, Florian
and Wang, Qingju",
editor="Bertoni, Guido
and Coron, Jean-S{\'e}bastien",
title="Fides: {L}ightweight {A}uthenticated Cipher with {S}ide-{C}hannel {R}esistance for {C}onstrained {H}ardware",
booktitle="Cryptographic Hardware and Embedded Systems - CHES 2013",
year="2013",
publisher="Springer Berlin Heidelberg",
address="Berlin, Heidelberg",
pages="142--158",
isbn="978-3-642-40349-1",
url= {https://doi.org/10.1007/978-3-642-40349-1_9},
DOI= {10.1007/978-3-642-40349-1_9}
}

@InProceedings{PRINCE,
  author=       "Borghoff, Julia
                and Canteaut, Anne
                and G{\"u}neysu, Tim
                and Kavun, Elif Bilge
                and Knezevic, Miroslav
                and Knudsen, Lars R.
                and Leander, Gregor
                and Nikov, Ventzislav
                and Paar, Christof
                and Rechberger, Christian
                and Rombouts, Peter
                and Thomsen, S{\o}ren S.
                and Yal{\c{c}}{\i}n, Tolga",
  editor=      "Wang, Xiaoyun and Sako, Kazue",
  title=       "{PRINCE} -- {A} {L}ow-{L}atency {B}lock {C}ipher for {P}ervasive {C}omputing {A}pplications",
  booktitle=    "Advances in Cryptology -- ASIACRYPT 2012",
  year=         "2012",
  publisher=    "Springer Berlin Heidelberg",
  address=      "Berlin, Heidelberg",
  pages=        "208--225",
  url= {https://doi.org/10.1007/978-3-642-34961-4_14},
  DOI= {10.1007/978-3-642-34961-4_14}
}

@InProceedings{SKINNY,
  author=         "Beierle, Christof
                  and Jean, J{\'e}r{\'e}my
                  and K{\"o}lbl, Stefan
                  and Leander, Gregor
                  and Moradi, Amir
                  and Peyrin, Thomas
                  and Sasaki, Yu
                  and Sasdrich, Pascal
                  and Sim, Siang Meng",
  editor=         "Robshaw, Matthew and Katz, Jonathan",
  title=          "The {SKINNY} {F}amily of {B}lock {C}iphers and {I}ts {L}ow-{L}atency {V}ariant {MANTIS}",
  booktitle=      "Advances in Cryptology -- CRYPTO 2016",
  year=           "2016",
  publisher=      "Springer Berlin Heidelberg",
  address=        "Berlin, Heidelberg",
  pages=          "123--153",
  url= {https://doi.org/10.1007/978-3-662-53008-5_5},
  DOI= {10.1007/978-3-662-53008-5_5}
  }

@InProceedings{MIDORI,
author=         "Banik, Subhadeep
                and Bogdanov, Andrey
                and Isobe, Takanori
                and Shibutani, Kyoji
                and Hiwatari, Harunaga
                and Akishita, Toru
                and Regazzoni, Francesco",
                editor="Iwata, Tetsu
                and Cheon, Jung Hee",
title=          "Midori: {A} {B}lock {C}ipher for {L}ow {E}nergy",
booktitle=      "Advances in Cryptology -- ASIACRYPT 2015",
year=           "2015",
publisher=      "Springer Berlin Heidelberg",
address=        "Berlin, Heidelberg",
pages=          "411--436",
url= {https://doi.org/10.1007/978-3-662-48800-3_17},
DOI={10.1007/978-3-662-48800-3_17}
}

@misc{Yan_RD_2026,
    author = {Xueping Yan and Lin Tan and Wenfeng Qi},
    title = {Related-Differential Distinguishers on up to 7 Rounds of {AES}},
    howpublished = {Cryptology {ePrint} Archive, Paper 2026/1039},
    year = {2026},
    url = {https://eprint.iacr.org/2026/1039}
}

% Author biographies

\end{document}